 \documentclass[onecolumn,journal]{IEEEtran} 
 \pdfoutput=1
 
 \IEEEoverridecommandlockouts

\usepackage{amsfonts,mathrsfs}
\usepackage{amssymb}
\usepackage{url,color}
\usepackage{subfigure}
\usepackage{algpseudocode}
\usepackage{algorithm}
\usepackage{makeidx}         
\usepackage{graphicx}

\usepackage[top=1.2in, bottom=1.2in, left=1.5in, right=1.5in]{geometry}

\usepackage{cite}

\ifCLASSINFOpdf
\else
\fi
\usepackage[cmex10]{amsmath}

\usepackage{cases}
\usepackage{url}

\begin{document}
%
\title{Achieving the Optimal Steaming Capacity and Delay Using  Random Regular Digraphs in P2P Networks\thanks{The research reported here was funded by NSF grant CNS-0964081 and AFOSR MURI FA 9550-10-1-0573.}}

\author{Joohwan Kim and R. Srikant\\
Coordinated Science Lab and\\
Department of Electrical and Computer Engineering\\
University of Illinois at Urbana-Champaign\\
joohwan@illinois.edu; rsrikant@illinois.edu}



%
\author{\IEEEauthorblockN{Joohwan Kim\IEEEauthorrefmark{1} and
R. Srikant\IEEEauthorrefmark{2}}\\
\IEEEauthorblockA{Coordinated Science Lab and\\
Dept. of Electrical and Computer Engineering\\
University of Illinois at Urbana-Champaign\\
Email: \IEEEauthorrefmark{1}joohwan@illinois.edu; \IEEEauthorrefmark{2}rsrikant@illinois.edu}}


\maketitle

\newtheorem{definition}{Definition}{\bf}{}
\newtheorem{corollary}{Corollary}{\bf}{}
\newtheorem{proposition}{Proposition}{\bf}{}
\newtheorem{lemma}{Lemma}{\bf}{}
\newtheorem{theorem}{Theorem}{\bf}{}
\newtheorem{remark}{Remark}{\bf}{}
\newtheorem{assumption}{Assumption}{\bf}{}

\def\re{\mathbb{R}}
\def\inc{{\rm inc}}

\def\an#1{{\color{red}#1}}

\newcommand{\x}[2]{S_{#2}^{(#1)}}
\newcommand{\z}[2]{Z_{#2}^{(#1)}}

\newcommand{\cA}[1]{\mathcal{A}_{#1}}
\newcommand{\cB}[1]{\mathcal{B}_{#1}}

\newcommand{\setN}[1]{\mathcal{N}_{#1}}
\newcommand{\setS}[1]{\mathcal{S}_{#1}}
\newcommand{\SC}[1]{S_{>#1}}
\newcommand{\Sd}[1]{S_{#1}}
\newcommand{\Shone}{S_{h+1}}
\newcommand{\Sh}{S_h}

\newcommand{\N}[1]{N_{#1}}
\newcommand{\NC}[1]{N_{> #1}}
\newcommand{\Nd}{N_{d}}
\newcommand{\Ndone}{N_{d+1}}

\newcommand{\comb}[2]{{#1 \choose #2}}

\newcommand{\dmin}[1]{\delta_{\min, #1}(\epsilon)}
\newcommand{\dmax}[1]{\delta_{\max, #1}(\epsilon)}
\newcommand\prob[1]{P\left[#1 \right]}
\newcommand{\NN}{\nonumber}
\newcommand{\odd}{\text{odd}}
\newcommand{\even}{\text{even}}
\newcommand{\cH}{\mathcal{H}}
\newcommand{\cEodd}{\mathcal{E}_{\odd}}
\newcommand{\cEeven}{\mathcal{E}_{\even}}
\newcommand{\cVodd}{\mathcal{V}_{\odd}}
\newcommand{\cVeven}{\mathcal{V}_{\even}}

\newcommand{\cG}{\mathcal{G}}
\newcommand{\cV}{\mathcal{V}}
\newcommand{\setX}{\mathcal{X}}
\newcommand{\setY}{\mathcal{Y}}
\newcommand{\setR}{\mathcal{R}}

\newcommand{\tblue}{\textcolor{black}}
\newcommand{\tred}[1]{}

\newcommand{\tech}{\textcolor{black} }

\newcommand{\jour}{}

\newlength{\mylength}
\newlength{\mylengthNarrow}
\newlength{\mylengthVeryNarrow}

\setlength{\mylength}{5in}
\setlength{\mylengthNarrow}{4in}
\setlength{\mylengthVeryNarrow}{3in}


\begin{abstract}
In earlier work, we showed that it is possible to achieve $O(\log N)$ streaming delay with high probability in a peer-to-peer network, where each peer has as little as four neighbors, while achieving any arbitrary fraction of the maximum possible streaming rate. However, the constant in the $O(log N)$ delay term becomes rather large as we get closer to the maximum streaming rate. In this paper, we design an alternative pairing and chunk dissemination algorithm that allows us to transmit at the maximum streaming rate while ensuring that all, but a negligible fraction of the peers, receive the data stream with $O(\log N)$ delay with high probability. The result is established by examining the properties of graph formed by the union of two or more random 1-regular digraphs, i.e., directed graphs in which each node has an incoming and an outgoing node degree both equal to one.
\end{abstract}


%
\IEEEpeerreviewmaketitle

\section{Introduction}\label{sec: intro}

Consider $N$ nodes in a network with no pre-defined edges. Draw a directed Hamiltonian cycle among these nodes, i.e., a directed cycle which uses all nodes, where the Hamiltonian cycle is picked uniformly at random from the set of all possible Hamiltonian cycles. Now draw another random directed Hamilton cycle through the same set of nodes. Even though each cycle has a diameter of $N,$ the surprising result in \cite{Kim01} states that the graph formed by the union of the two cycles (which we will call the \emph{superposed graph}) has a diameter of $O(\log N)$ with high probability. In \cite{Law03}, it was shown that random Hamiltonian cycles can be formed among peers (now the peers are the nodes mentioned above) in a peer-to-peer (P2P) network by exploiting the natural churn (i.e., arrivals and departures of peers) in the network. Next, suppose that each edge is assumed to have a unit capacity  and one of the nodes in the superposed graph is designated as a source node, it is not difficult to show that the cut capacity of the superposed graph is 2, i.e., the minimum capacity of any cut that separates the source from some of the other peers is equal to 2. A fundamental result in graph theory
\cite{Chu65,Edmonds1972,Tarjan77} asserts that the cut capacity can be achieved by constructing 2 edge-disjoint arborescences (for readers unfamiliar with the term, an arborescence is a directed spanning tree) and transmitting data over each arborescence at the unit rate, such that the total rate is equal to the cut capacity. However, how to construct such arborescences in a distributed fashion is an open question in general. Further, if the superposed graph is used to transmit real-time data in a P2P network, then delay becomes an important consideration, but delay is not addressed in the papers mentioned in this paragraph.

In \cite{Kim13TiT}, we presented an algorithm which achieves any arbitrary fraction of the cut capacity of the superposed graph while the maximum delay in the graph scales as $O(\log N)$ with high probability. The result is established by considering the following model of a random graph: consider the superposed random graph mentioned previously and suppose that each edge in the second cycle  is removed independently with  probability $1-1/K$, the remaining graph continues to have a diameter of $O(\log N/\log (1+2/K))$ with high probability. Based on this theoretical result, we showed in \cite{Kim13TiT} that  the proposed algorithm can achieve a $1-1/K$ fraction of the optimal streaming rate with $O(\log N/\log (1+2/K))$ delay. A caveat with this result is that the delay increases without bound as we decreases $q$ to get a near-optimal streaming rate. Hence, we have to compromise a fraction of the optimal streaming capacity for $O(\log N)$ delay.  Practical implementation \cite{Kim13tech3} indicates that this fraction can be significant to achieve reasonable delays.

In this paper, we address the above shortcoming of the previous algorithm and present a new algorithm which achieves the maximum streaming capacity. Our new algorithm is able to stream data to all peers, with the possible exception of an $o(1)$ fraction of them, at the maximum possible rate with $O(\log N)$ delay. In particular, the algorithm does not need to compromise any fraction of the optimal streaming rate for the $O(\log N)$ delay. Another way to view the result is that our algorithm constructs edge-disjoint arborescences   in a distributed manner to achieve the cut capacity of the superposed graph, to practically all the nodes in the network, in the limit as $N$ goes to infinity. Our algorithm is as simple as the previous one and continues to be robust to peer churn.

We now briefly comment on the relationship between our work and other prior work in the literature, primarily concentrating on those papers that establish fundamental performance bounds. The capacity of real-time, streaming node-capacitated P2P networks was studied in \cite{Kumar07,Mundinger08}. Structured approaches to achieving this capacity have been studied in \cite{Li05,Liu08,Liu10}, where typically multiple trees are constructed to carry the streaming traffic. An important aspect of this style of research is to design algorithms to maintain the tree structure even in the presence of peer arrivals and departures. Additionally, in this literature, trees are often constructed to have small depths to reduce streaming delay, while achieving full streaming capacity. We note that churn is not a problem for us since the Hamiltonian cycle construction takes advantage of peer churn \cite{Law03,Kim13TiT}. In comparison with prior work on structured networks, the main contribution of our paper is to establish the capacity and delay result for a network which handles churn naturally. The other style of research is to study unstructured P2P networks \cite{Frieze1985,Massoulie07, Sanghavi07,Bonald08}. The main idea in these papers is to consider gossip-style data dissemination in fully-connected networks and show that either the full capacity or a fraction of the capacity is achieved with small delay. In contrast, in our model, we show that the network can achieve full capacity and small delay even when each peer has only four neighbors (two upstream and two downstream). The advantage of having small neighborhood size is due to that fact that storing and updating neighborhood information in P2P networks incurs a significant overhead. An exception to the fully-connected network model in prior work can be found in \cite{Zhao11}; however, only a capacity result is established in this paper but no algorithm is provided to achieve the capacity, and furthermore, a delay bound is not provided either. Motivated by the widespread use of CoolStreaming \cite{Zhang05coolstreaming}, delay bounds based on mean-field-type approximations have been studied in \cite{Zhou07,Shakkottai11}, but the mean-field limit obscures network construction details in these models and hence are not suitable to study network topologies. In addition to the theoretical results mentioned here, a number of other algorithms have been proposed and studied using simulations. We refer the reader to \cite{Liu07opportunitiesand,Dongni08,Dan09,Traverso12} and references within. The rest of the paper is organized as follows. We first present our network model and main result, and provide some intuition behind the main result. Following that, we present the key lemmas and theorems that lead up to the main result, relegating the proofs to the appendix. We wrap up the paper with concluding remarks, briefly commenting on ongoing work to convert the theory here to a practical implementation. This paper focuses only on the theory behind our design. A number of practical issues must be addressed in a real implementation. We have implemented our P2P algorithm in a testbed, the details of which can be found in \cite{Kim13tech3}.

\section{System Model}

The streaming network is assumed to consist of a source peer that generates chunks from a real-time video stream and regular peers that  receive the chunks to play the video.  The network operation is controlled
by two algorithms: the peer-pairing algorithm and the chunk-dissemination algorithm. The former is for constructing and maintaining the network topology when peers join and leave frequently. The latter is for disseminating chunks generated at the source to all peers over the topology maintained by the former algorithm.

\subsection{Peer-Pairing Algorithm}

The peer-pairing algorithm constructs $M$ independent 1-regular digraphs. The topology of the P2P system is the union of these $M$ digraphs. For the theoretical results in the paper, $M$ can be as little as two, but can be larger for robustness reasons in a practical implementation. Since every peer has one  incoming edge and one outgoing edge in a 1-regular digraph, it has $M$ incoming edges and $M$ outgoing edges, and thus the resulting topology is a $M$-regular digraph. We index the $M$ incoming edges and the $M$ outgoing edges from 1 to $M$, respectively. We call the peer at the other end of the $m$-th incoming edge of peer $v$  \emph{the $m$-th parent} of peer $v$ and denote by $p_m(v)$. Similarly, we call the peer in the other  end of the $m$-th outgoing edge of peer $v$ \emph{the $m$-th child} and denote by $c_m(v)$. Peer $v$ receives chunks from the parents through its incoming edges and send chunks to the  children through its outgoing edges.

\begin{algorithm}
\caption{Peer-Pairing Algorithm}\label{peerAlgo}
\begin{enumerate}
\item Initially, the source (peer 1) creates $M$ loops, i.e.,
$p_m(1)=c_m(1)=1$ for all $m$. The set $V$ of peers is given by $\{1\}$.

\item When a new peer $v$ joins the network, the  set $V$ of existing peers is updated as
\begin{equation}
V\leftarrow V\cup\{v\}.\NN
\end{equation}
$M$ peers $w_1,\cdots, w_M$ are chosen from $V$ allowing repetition. For each $m$, if $w_m=v$, peer $v$ creates a loop, i.e., $p_m(v)=v$ and $c_m(v)=v$. If $w_m\neq v$, peer $v$ contacts peer $w_m$ and breaks into the $m$-th outgoing edge $(w_m, w_m')$ where $w_m'=c_m(w_m)$ as follows:
$$p_m(v)\leftarrow w_m, c_m(v)\leftarrow w_m',$$
$$c_m(w_m)\leftarrow v, p_m( w_m' ) \leftarrow  v.$$
As a result, existing  edge $(w_m, w_m')$ is replaced with $(w_m, v)$ and $(v, w_m')$ for each $m$.

\item When an existing peer $v$ leaves, its child $c_m(v)$ contacts its parent $p_m(v)$ and establishes a new edge between both. As a result, exiting edges $(p_m(v), v)$ and $(v, c_m(v))$ are replaced with $(p_m(v),c_m(v))$.
\end{enumerate}
\end{algorithm}

We  present the peer pairing algorithm, i.e., the algorithm used by every peer  $v$ to find its  parents and children,  in Algorithm~\ref{peerAlgo}. At a given time instant, let $L_m=(V,E_m)$ denote a digraph with all peers and their $m$-th outgoing edges, i.e., $E_m=\{(v, c_m(v))| v\in V\}$. We call this graph \emph{layer $m$} for $m=1,2,\cdots, M$.
The network topology  is then  the graph formed by the superposition   of  all the layers. A layer begins with an initial peer, the source, with a loop. When a new peer $v$ joins, it chooses an edge from the layer and breaks into the middle of the edge or creates a loop. When an existing peer leaves, its incoming edge and outgoing edge in the layer are connected directly. Hence, at any given movement, a layer must be a 1-regular digraph where every peer has only one incoming edge and one outgoing edge in the layer.

Let $\Pi(V)$ be the set of all possible 1-regular digraphs that can be made of peers in $V$. A simple way to construct a 1-regular digraph is by a permutation of the peers.
For $V=\{1,\cdots, N\}$, let $(v^{(1)}, v^{(2)}, \cdots, v^{(N)})$ be a permutation of $(1,2, \cdots, N)$. If we draw an  edge from $i$ to $v^{(i)}$ for each $i$, all peers will have exactly one incoming edge and one outgoing edge. Hence, the resulting graph becomes 1-regular digraph. Since the number of all possible permutations is $N!$, we can make  $N!$ different 1-regular digraphs using  $N$ peers, i.e., $\Pi(V)=|V|!$. The following result shows that each layer is one of these graphs equally likely and is independent of other layers:
\begin{proposition}
At a given time instance, suppose $V$ is the set of existing peers. Then, layers $L_1$, $L_2$, ..., $L_M$ are mutually independent and uniformly distributed in $\Pi(V)$, i.e., for each $G\in \Pi(V)$
$$\prob{G_m=G}=\frac{1}{|V|!} \;\;\text{for each }m.$$
\end{proposition}

Note the key difference of the peer-pairing algorithm in this paper from that in our previous paper \cite{Kim13TiT}. In the previous algorithm, the  name server first  chooses $(w_1, \cdots, w_M)$ from $V$ and then adds new peer $v$ to $V$, which eliminates the probability of having  $w_m=v$. Hence, new  peers do not create a loop and thus each layer is a Hamiltonian cycle, which is a special case of regular-1 digraphs (see  \cite{Kim13TiT} for the details). In contrast, new peers in this paper may create loops in a layer and thus a layer could have multiple cycles or loops.

Although the algorithm constructs and maintains the network topology in a distributed fashion, it still needs a name server that provides a new peer with the IP addresses of $M$ existing peers. The server performing this minimal function is called a tracker, and is commonly used in most P2P networks. Under our algorithm, the name server only  knows the set of existing peers and does not need to know the connections between them. Further, the set of existing peers could be outdated, in which case new peers contact the name server again if some of the retrieved peers are not reachable. Lastly, our algorithm is scalable because every peer exchanges chunks with $2M$ peers, regardless of the network size.

\subsection{Chunk Dissemination Algorithm}

We next present the chunk-exchange algorithm that  disseminates chunks over the network topology constructed by the aforementioned peer-pairing algorithm. We assume that the network is time-slotted and the source  generates up to $M$ unit-size chunks from real-time video contents during every time slot. We also assume that  every peer can upload at most one chunk through every outgoing edge. Since every peer has $M$ outgoing edges, we are assuming that every peer has an upload bandwidth of $M$. Since every peer has $M$ incoming edges, the download bandwidth required for our algorithm is also $M$. When a peer receives a chunk at a give time slot, this peer can transmit the chunk from the next time slot.

\begin{algorithm}
\caption{Random Flow-Assignment (RFA) Algorithm}\label{coloringAlgo}
\begin{algorithmic}
\State $d(v)\leftarrow\infty$, $\chi(v)\leftarrow 0$,
\State $d^*\leftarrow \lceil \log_2\frac{N}{\log^{c}N} \rceil$ for $0<c<1$.
\State $\mathcal{M}=\mathcal{F}=\{1,2,\cdots,M\}$.
\If {there exists $m^*$ such that $p_{m^*}(v)=1$}
	\State $d(v)\leftarrow 1$, $\chi(v)\leftarrow m^*,$ and $f_{m^*}(v)\leftarrow m^*.$
	\State $\mathcal{M}\leftarrow \mathcal{M}\setminus \{m^*\}$,
			$\mathcal{F}\leftarrow \mathcal{F}\setminus\{f_{m^*}(v)\}$.
\ElsIf {there exists a parent $p_m(v)$ with $d(p_m(v))<d^*$}
	\State $m^*\leftarrow \mathop{\arg\min}_m d(p_m(v))$, where ties are broken randomly.
	\State $d(v)\leftarrow  d(p_{m^*}(v))+1$, $\chi(v)\leftarrow \chi(p_{m^*}(v)),$
	\State and $f_{m^*}(v)\leftarrow \chi(v)$	
		\State $\mathcal{M}\leftarrow \mathcal{M}\setminus \{m^*\}$,
				$\mathcal{F}\leftarrow \mathcal{F}\setminus\{f_{m^*}(v)\}$.
\EndIf
\For {each $m$ in $\mathcal{M}$}
	\State randomly choose $f_m(v)$ from $\mathcal{F}$.
			\State $\mathcal{F}\leftarrow \mathcal{F}\setminus\{f_{m}(v)\}$.
\EndFor
\State run the algorithm again if a parent $i$ changes $d(i)$ or $\chi(i)$.
\end{algorithmic}
\end{algorithm}

For chunk dissemination, we split the chunks into $M$ flows, named flow~1, flow~2, ...,  flow~$M$. When the source generates $M$ or less chunks  at a given time slot, it assigns the $M$ flows to the chunks uniquely. We call the chunk assigned flow-$f$ \emph{flow-$f$ chunks}. Every peer also assigns $M$ flows to its $M$ incoming edges uniquely. As a result, every edge is assigned one of the $M$ flows. We call edges assigned flow-$f$ \emph{flow-$f$ edges}.
Let $f_m(v)$ be the flow assigned to the $m$-th incoming edge of peer $v$. Every peer $v$ determines  $f_m(v)$ for each $m$ using the flow-assignment algorithm in Algorithm~\ref{coloringAlgo}. Under the flow-assignment algorithm, every peer assigns $M$ flows to all the incoming edges uniquely while a peer may have two or more outgoing edges with the same flow.

Under our chunk dissemination algorithm, each chunk from a flow  is transmitted through the same-flow edges.  Specifically, we assume that every peer maintains a FIFO (First-In-First-Out) queue for each outgoing edge. When a peer receives a flow-$f$ chunk from a parent or the source peer generates a flow-$f$ chunk, the peer  adds the chunk to the queues associated  with the flow-$f$ outgoing edges. At the beginning of a time slot, every peer takes the oldest chunk from each queue and transmits the chunk through the outgoing edge associated with the queue.

\begin{figure}
\centering
\includegraphics[width=0.95\mylength]{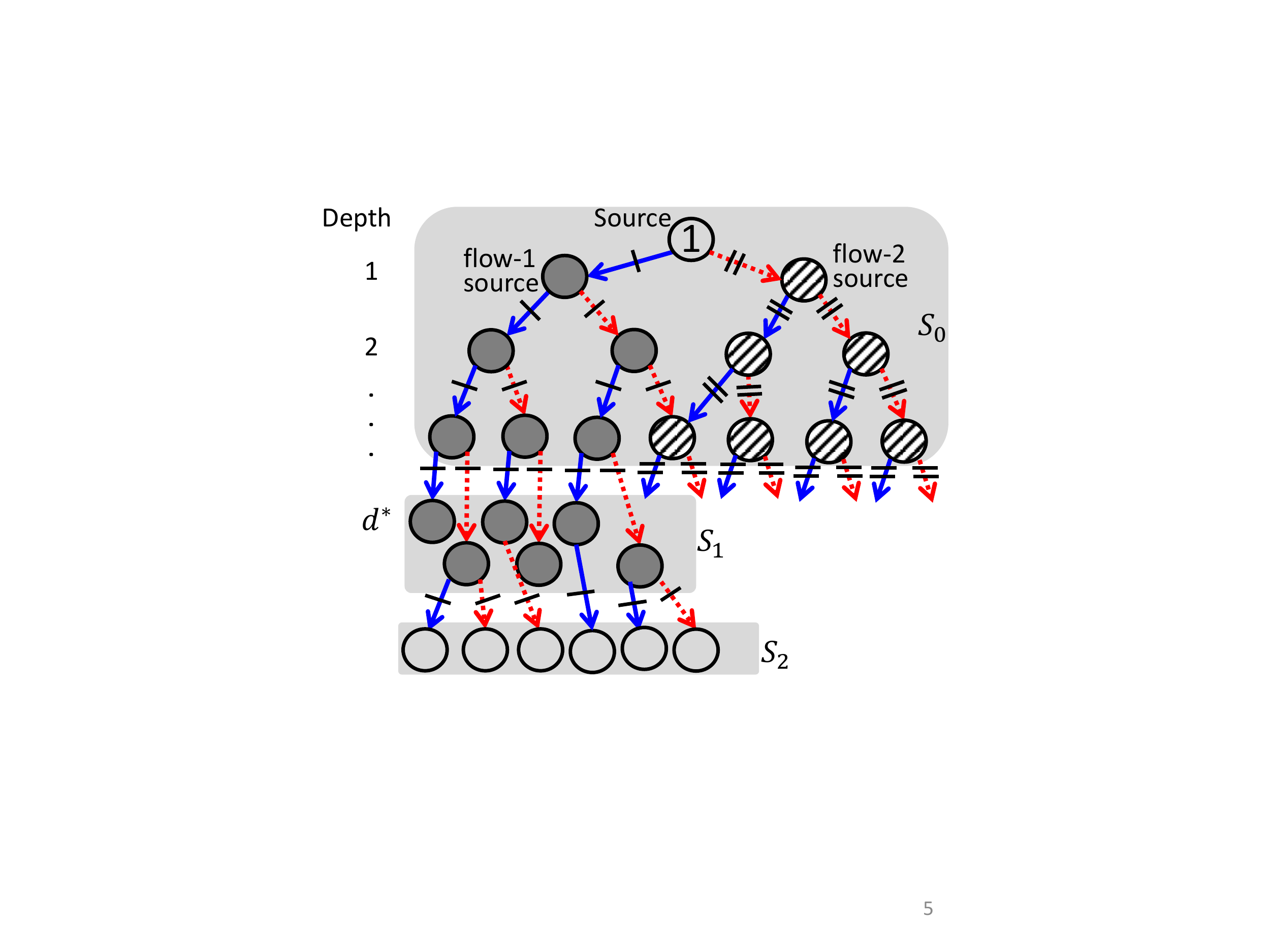}
\caption{Flow assignment under the RFA algorithm ($M=2$): Solid lines represent the edges in layer~1 and dotted lines represent the edges in layer~2. Edges with ``$\;|\;$'' in the middle represent flow-1 edges while edges with ``$\;||\;$''  represents flow-2 edges. Among the peers $v$ with $d(v)\leq d^*$, the gray peers represents flow-1 peers, the peers $v$ with $\chi(v)=1$, and the shaded peers represents flow-2 peers, the peers $v$ with $\chi(v)=2$. The flow-2 peers in depth $d^*$ are not drawn in the figure.}
\label{fig:algorithmDescription}
\end{figure}

Our chunk dissemination algorithm called RFA is presented in Algorithm 2. We next describe how the RFA algorithm works. See Fig.~\ref{fig:algorithmDescription} that shows flow assignment for the simplest case $M=2$. Under the algorithm, the $m$-th child of the source assigns flow-$m$ to its $m$ incoming edge from the source. Hence, the source sends flow-$m$ chunks through its $m$-th outgoing  edge. Since the $m$-th child is the  peer that directly receives flow-$m$ chunks from the source, we call it \emph{flow-$m$ source.} The $m$-th child then sets  its distance to be 1, i.e., $d(v)=1$, and $\chi(v)=m$.
We call the peers with $d(v)\leq d^*$ and $\chi(v)=m$  \emph{flow-m peers} because flow-m chunks are disseminated through these peers.
The children $v$ of the flow-1 source will choose $d(v)=2$, and
  $\chi(v)=1$. $d(v)=2$ means the distance of peer $v$ from the source is 2, and  $\chi(v)=1$ implies that peer $v$ is closer to the flow-1 source than the flow-2 source. This holds for $d(v)=1,2,\cdots, d^*$, for an appropriately chosen $d^*$ which will be described later.
 Note that if a peer has the same distance from the flow-1 source or the flow-2 source and has two parents with different main flows (as the forth peer in depth 3 in Fig.~\ref{fig:algorithmDescription}) it decides its main flow randomly and then the children will be affected by the decision. Beyond depth $d^*$, peers assign flow~1 to one of its incoming edges randomly regardless of the depth or the main flow of its parents.

\section{Streaming Rate and Delay}

In this section, we analyze the throughput and delay performance of the proposed algorithms. We first introduce \emph{flow graphs} that show the dissemination of each flow and then show how to relate these flow graphs to the throughput and delay performance.

In order to analyze the dissemination of flow $f$, we need to characterize a  graph consisting only of flow-$f$ edges. Define \emph{flow graph $f$} to be  $F_f\triangleq (V, \cup_{m=1}^M E_m^{(f)})$ where
$$E_m^{(f)}=\{(v',v )\in E_m| f_m(v)=f   \}.$$ Since every peer assigns flows to its incoming edges using one-to-one mapping, every peer has exactly one incoming edge in each flow graph, which leads to the following property:
\begin{lemma}\label{lemma:flow_graph}
In flow graph $f$, every peer $v$ excluding the source satisfies only one of the following:
\begin{enumerate}
\item There is a unique path from the source to peer $v$
\item There is no path from the source to peer $v$.
\end{enumerate}
\end{lemma}
\begin{proof}
Since every peer has only one incoming edge, we trace backward the incoming edges from a peer $v$. This backward tracing will end when we hit the source because the source does not have any incoming edge in the flow graph. If we do not hit the source, we should revisit a peer twice.  Hence, there is a unique path from the source to peer $v$ or there is a cycle where peer $v$ is connected.
\end{proof}

\begin{figure}
\centering
\includegraphics[width=0.4\mylength]{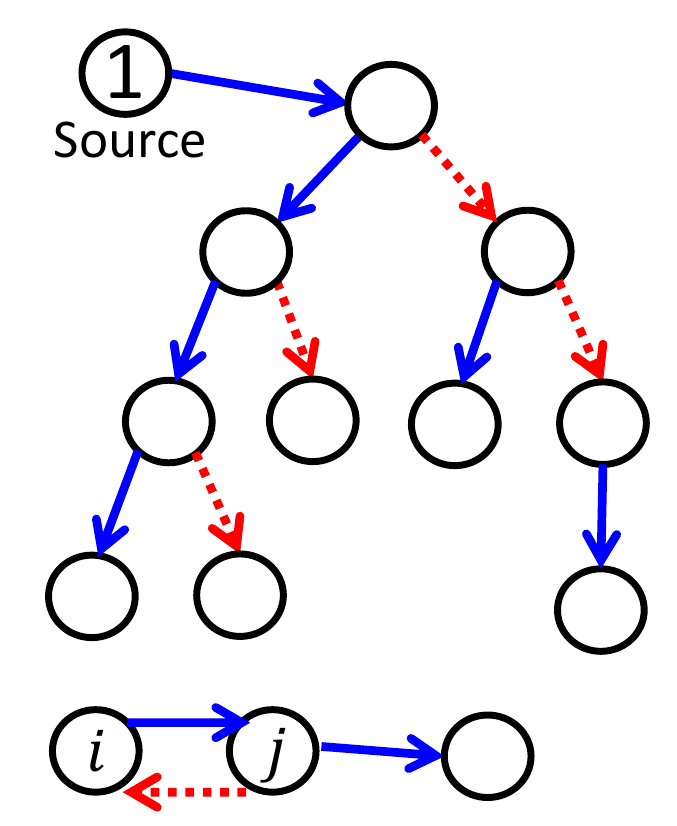}
\caption{An example of flow graph 1: Solid  edges represent the edges in layer 1 while dotted edges represent the edges in layer 2. All the edges in the figure are flow-1 edges. Peers belong to  either a tree rooted at the source or a cycle (including a loop).}
\label{fig:flowgraph}
\end{figure}

An example of a flow graph is illustrated in Fig.~\ref{fig:flowgraph}.
From this lemma, peers are either connected or disconnected from the source. Since the path from the source to a connected peer is unique, there must be a tree that is rooted at the source and covers all the connected peers. The lemma also implies that the disconnected peers form a loop or  a cycle among themselves. Since flow-$f$ chunks are transmitted only through flow-$f$ edges, the disconnected peers  cannot receive the chunks. In Section~\ref{sec:simulationResults}, we proposed an extension of our algorithms that can disseminate the missing flows to the disconnected peers.

The next lemma shows whether the connected peers can receive all the flow-$f$ chunks or not:
\begin{lemma}\label{lemma:dissemination_delay}
A flow-$f$ chunk generated by the source  at time slot $t$ is received by peer $v$ by time slot $t+h$, where $h$ is the distance from the source to peer $v$ in flow graph $f$.
\end{lemma}
\begin{proof}
From Lemma~\ref{lemma:flow_graph}, if peer $v$ belongs to a cycle, it should be disconnected from the source in flow graph $f$, i.e., $h=\infty$. Hence, this lemma is satisfied trivially.
If peer $v$ belongs to the tree rooted at the source, there is a unique  path from the source to peer $v$ that consists of $h$ edges. Since each edge is associated with a FIFO queue, the route can be seen as $h$ cascaded FIFO queues. Since every queue processes one chunk per time slot and at most one chunk arrives  in the first queue during a time slot, the peer at the end of the cascaded queue can receive a new chunk within $h$ time slots.
\end{proof}

We now analyze the throughput performance using  flow graphs. If a peer $v$ is  connected from a source in flow graph $f$, Lemma~\ref{lemma:dissemination_delay}
shows that the peer will eventually receive all flow-$f$ chunks. Hence, if a peer is connected from the source in all the flow graphs, the peer can receive all the chunks generated by the source and thus achieves the reception rate of $M$. It is shown in
 \cite{Edmonds1972} that the optimal streaming rate that can be guaranteed to all peers is $M$ when peers contribute an upload bandwidth $M$. Hence, we can say that the network is near optimal in throughput if almost all peers are connected from the source in all the flow graphs. We will show in the next section that $(1-o(1))$ fraction of peers satisfy this condition with high probability.

Besides the throughput performance, we are also interested in the time that it takes for each chunk to be delivered to almost all peers. Lemma~\ref{lemma:dissemination_delay} implies that a flow-$f$ chunk that is newly generated at time slot $t$ can be disseminated to all the peers within $h$ hops from the source in flow graph $f$ no later than time slot $t+h$.  In other words, if almost all peers are within $\Theta(\log N)$ hops from the source in each flow graph, we can say that these peers can receive all the chunks within $\Theta(\log N)$ time slots. Since $\Theta(\log N)$ is the optimal delay in P2P streaming, proving that implies that our algorithm achieves the optimal delay performance. Hence, in the next section, we focus on the depth of the tree in each flow graph.

\section{Depth of A Flow Graph}

In this section, we show that almost all peers are within $O(\log N)$ hops from the source in every flow graph with high probability. To show this, we first prove the result holds for flow graph~1, then extend the result to all flow graphs using a union bound. The proofs are fairly technical, but the key ideas are fairly easy to grasp. So in the interest of space and for simplicity of exposition, we present the key ideas in this section, and refer the reader to \cite{Kim13tech2} for the details of the proofs.

To show that almost all peers are within $\Theta(\log N)$ hops from the source in flow graph~1, we take the following three steps:

\textbf{Step 1:} We first show that the graph up to depth $d^*$, e.g.,  $S_0$ in Fig.~\ref{fig:algorithmDescription}, is a binary tree, except for a few peers in the middle  who may have only one outgoing edge in the tree, with high probability. A half of the tree consists only of flow-1 edges passing through flow-1 peers (e.g., the gray peers and their flow-1 edges in Fig~\ref{fig:algorithmDescription}). This means that  all flow-1 peers within depth $d^*$ are connected from the source within $d^*$ hops in flow graph~1.

\textbf{Step 2:} We show that almost all peers below depth $d^*$ are connected from the flow-1 peers in depth $d^*$, set $S_1$ in Fig.~\ref{fig:algorithmDescription}, through $\log N$ or less flow-1 edges. Hence, these peers are connected from the source within $O(\log N)$ hops in flow graph~1 with high probability.

\textbf{Step 3:} We finally show that almost all  flow-2 peers within depth $d^*$ have their flow-1 incoming edges from the peers below depth $d^*$ with high probability. Since almost all peers below depth $d$ are connected with $O(\log N)$ hops from the source in flow graph~1, so are the flow-2 peers.

\subsection{Expansion in the Shortest-Path Arborescence}

\begin{figure}[ht]
\centering
\subfigure[Before drawing edges from the peers in depth $2$]{\label{fig:randomEdges1}
\includegraphics[width=0.9\mylengthNarrow]{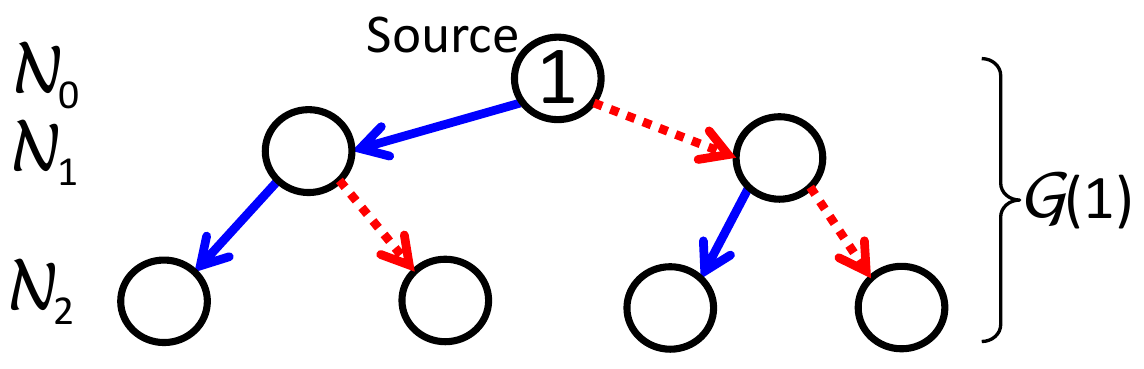}} 

\subfigure[After drawing layer-1 edges from $\setN{2}$: $\N{3}'=3$]{\label{fig:randomEdges2}
\includegraphics[width=0.9\mylengthNarrow]{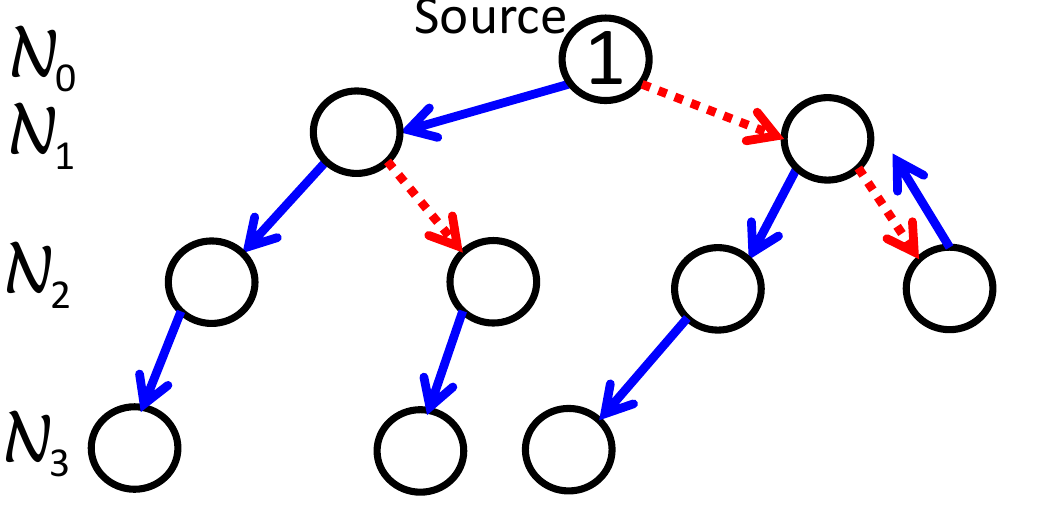}} 

\subfigure[After drawing layer-2 edges from $\setN{2}$: $\N{3}''=3$]{\label{fig:randomEdges3}
\includegraphics[width=0.9\mylengthNarrow]{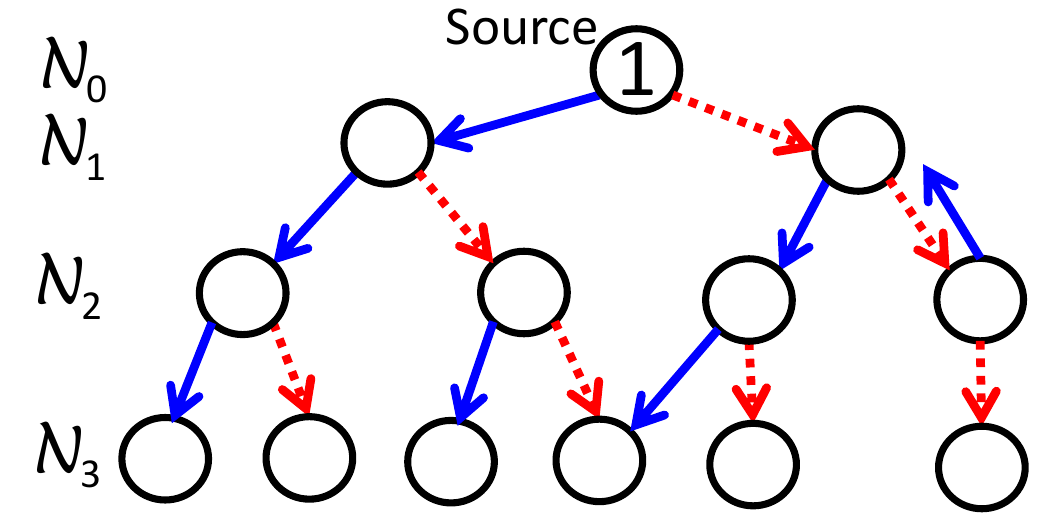}} 

\caption{Drawing random edges from the peers in depth 2:  solid edges represent layer-1 edges and dotted edges represent layer-2 edges. In this example, we have $\N{3}'=3$, $\N{3}''=3$, and thus $\N{3}=6$.}
\end{figure}

In this subsection, we first consider the number of peers in each depth $d\leq d^*$. Let $\setN{d}$ be the peers $v$ in depth $d$, i.e., $d(v)=d$, and let $\N{d}\triangleq |\setN{d}|.$  For simplicity, we also define
 $\setN{\leq d}\triangleq\cup_{h=0}^d \setN{d}$ and $\setN{>d}\triangleq \setN{\leq d}^C$ to be the set of peers that are within and beyond, respectively, depth $d$. Let  $\N{\leq d}=|\setN{\leq d}|$ and  $\N{> d}=|\setN{> d}|$.

Recall that each layer constructed by our peer-pairing algorithm is a random 1-regular digraph uniformly distributed among all possible 1-regular digraphs, and an equivalent graph can be constructed by random permutation. Specifically, for given peers and no initial  edges, we pick  the first peer  arbitrarily and draw an edge from the peer to a peer randomly chosen among the peers that do not have an incoming edge. Initially, all the peers can be a candidate because there is no edge drawn. We then pick the second peer among the peers that do not have an outgoing edge and draw an edge from the peer to a peer randomly chosen among the peers that do not have an incoming edge. If we draw $N$ edges repeating this method, the resulting graph becomes a random 1-regular digraph, which has the same distribution as an each layer.

Using this edge drawing  procedure with random permutation, we find the distribution of the number $\N{d}$ of peers in each depth $d$.  Define $\mathcal{G}(d)$ to be the graph of peers within depth $d+1$ and the outgoing edges from the peers within depth $d$. See the example in Fig.~\ref{fig:randomEdges1}. \emph{Note that when $\cG(d)$ is given, $\N{0}, \N{1}, \cdots, \N{d+1}$, $\N{\leq d+1}$, and $\N{>d+1}$ are determined.}
In the example in Fig.~\ref{fig:randomEdges1}, conditioned on $\cG(1)$, i.e., conditioned on that the network topology up to depth $2$ is given $\cG(1)$, we first draw layer-1 edges from the peers in depth $2$, i.e., $\setN{2}$. When we draw a layer-1 outgoing edge from the first peer in depth 2, there is one  layer-1 edge from each peer in $\setN{\leq 1}$. This means that there are $\N{\leq 1}$ peers that have a layer-1 incoming edge. Hence, the new edge can be drawn to one of the remaining $\N{>1}$ peers. If we draw an edge to one of the $\N{> 2}$ peers that are not within depth 2, the peer will be added to $\setN{3}$ as shown in Fig.~\ref{fig:randomEdges2}. If we draw an edge to one of the peers within depth $2$ (see the fourth peer in $\setN{2}$), the peer will not be added to $\setN{3}$. Hence, drawing layer-1 edges from the peers in $\setN{2}$ can be seen as the experiment where the peers in $\setN{2}$ choose one peer without repetition among $\N{>1}$ peers and a peer chosen from $\setN{>2}$ is considered to be a success.

Define $\setN{d}'$ to be the set of the peers in depth $d$ that are connected from the peers in depth $d-1$ through layer-1 edges. Let $\N{d}'=|\setN{d}'|$. In the example, $\N{3}'$ is equivalent to the number of red balls among $\N{2}$ balls drawn from a jar containing $\N{>1}$ balls including $\N{>2}$ red balls without replacing balls. The number of the red balls in this example is known to be a hyper geometric random variable with parameters $(\NC{1}, \NC{2}, \N{2})$. In general, conditioned on $\N{0}, \N{1},\cdots, \N{d}$,  $\N{d+1}'$ is a hyper geometric random variable with  $(\NC{d-1}, \NC{d}, \N{d})$ which has the following mean:
\begin{equation}
E[\N{d+1}'\;|\;\N{0},\cdots, \N{d}]= \frac{\NC{d}}{\NC{d-1}}\N{d}.\label{eq:mean_of_Nd'}
\end{equation}

After drawing the layer-1 edges from the peers in depth~2, we next draw layer-2 edges from the same  peers. Let $\setN{d}''$ be the set of the peers in depth $d$ that are connected from the peers in depth $d-1$ \emph{only} through layer-2 edges. Let $\N{d}''=|\setN{d}''|$. Since there is one  layer-2 edge from each  peers in $\setN{\leq 1}$, there are $\NC{1}$ peers that we can draw layer-2 edges to. If we drawn an edge from a peer in $\setN{2}$ to a peer not in $\setN{\leq 2}\cup \setN{3}'$, the peer is newly added to $\setN{3}''$, which will be considered to be a success. Otherwise, $\setN{3}''$ is not incremented. (See that the third peer in $\setN{3}$ in Fig~\ref{fig:randomEdges3} does not belong to $\setN{3}''$, but $\setN{3}'$.). Hence, drawing layer-2 edges from $\setN{2}$ can be seen as the experiment that $\N{2}$ peers are drawn without replacement from a jar containing $\N{>1}$ balls including $\N{>2}$ red balls. Hence, $\setN{3}''$ is equivalent to  the number of the red drawn balls in the experiment, which  is a hyper geometric random variable with $(\NC{1}, \NC{2}-\N{3}', \N{2})$. In general, $\N{d+1}''$ is a hyper geometric random variable with $(\NC{d-1}, \NC{d}-\N{d+1}', \N{d})$, which has the following mean:
\begin{equation}E[\N{d+1}''\;|\;\N{0},\cdots, \N{d}, \N{d+1}']= \frac{\NC{d}-\N{d+1}'}{\NC{d-1}}\N{d}.\label{eq:mean_of_Nd''}
\end{equation}
Combining (\ref{eq:mean_of_Nd'}) and (\ref{eq:mean_of_Nd''}), the mean of the number of the peers in depth $d+1$ is given by
\begin{equation}E[\N{d+1}\;|\;\N{0},\cdots, \N{d}]= 2\frac{\NC{d}\N{d}}{\NC{d-1}}-\frac{\NC{d}\N{d}^2}{(\NC{d-1})^2}.\label{eq:mean_of_Nd}
\end{equation}
From (\ref{eq:mean_of_Nd}), when $\N{\leq d}$ is small, i.e., the number of the peers within depth $d$ is small compared to $N$, the number $\Nd$ of the peers in depth $d$ is expected to double approximately as $d$ increases. We want to show that this phenomenon continues to hold for all $d\leq d^*$, which  requires a stronger concentration result.

 Using a property of a hyper geometric distribution, we derive the concentration result for pretty small $\Nd$:
\begin{lemma}\label{lemma:2-ary}
The shortest path arborescence up to depth $d$ is a binary tree  with probability
$$\prob{\N{h}=2\N{h-1}, \forall h\leq d}\geq 1-\frac{d\cdot 2^{2d}}{N-2^d}$$
\end{lemma}
\jour{The proof is provided in Appendix~A in our online technical report \cite{Kim13tech2}.}
\tech{The proof is provided in Appendix~\ref{append:lemma:2-ary}.} The lemma implies that the arborescence is 2-ary up to depth
$\frac{1}{2}(1-\epsilon) \log_2 N$ with high probability.
However, this concentration result does not hold for a higher depth. We can derive another concentration result for this case:
\begin{proposition}\label{prop:concentrationOnNd}
Conditioned on $\N{0}, \N{1},\cdots, \N{d}$, for $0<\phi<1-\frac{2\N{d}}{\NC{d-1}}$,
$$\prob{\N{d+1}<  2\phi \N{d}}\leq 2 \exp\left(-2 \sigma_d^2\N{d}\right)$$
where $\sigma_d=(1-\frac{2\N{d}}{\NC{d-1}}-\phi)$.
\end{proposition}
\jour{The detailed proof is provided in Appendix~B in \cite{Kim13tech2}.}
\tech{The detailed proof is provided in Appendix~\ref{appex:prop:concentrationOnNd}.}

\begin{proposition}\label{prop:mainExpansion}
Define
$$d_1=\left\lceil\frac{\log_2 N}{3} \right\rceil,\;\;d_2=\left\lceil \frac{5\log_2 N}{6}\right\rceil,$$
and
\begin{numcases}{\phi_d=}
1  & $0\leq d < d_1$,\NN\\
1-N^{-\frac{1}{9}} & $d_1\leq d < d_2$,\NN\\
\left(1-\frac{\Nd}{\NC{d-1}} \right)^3 & $d_2\leq d<d^*$.\NN
\end{numcases}
For any $\epsilon>0$, there exists $N_0$ such that for all $N>N_0$,
\begin{align}
&\prob{\N{d^*}\geq\frac{(1-\epsilon)N}{\log^c N},
\N{d}\geq 2 \phi_{d-1}\N{d-1}, \forall d\leq d^*}\NN\\
&\geq 1-\frac{2(1+\epsilon)\log_2 N}{N^{1/3}}.\NN
\end{align}
\end{proposition}
\jour{The proof is provided in Appendix~C in \cite{Kim13tech2}.}
\tech{The proof is provided in Appendix~\ref{append:prop:mainExpansion}.}

Proposition~\ref{prop:mainExpansion} implies that
the number of the peers in each depth $d$ doubles up to depth $d_1$ with probability 1. For  depth $d_1< d\leq d^*$, the number of the peers in depth $d$ increases by $2\phi_d$ times with high probability. Since $2\phi_d$ in this case is nearly 2, the arborescence up to depth $d^*$ is  nearly a binary tree with high probability.

We next consider the number of flow-1 peers in each depth. Recall that under the RFA algorithm, every peer $v$ within depth $d^*$ determines its main flow $\xi(v)$.  In Fig.~\ref{fig:algorithmDescription},  flow-1 peers are shown in gray. Let $\setX_d$ be the set of gray peers in depth $d$, and $X_d=|\setX_d|$. By symmetry, the expected number of flow-1 peers in depth $d$ must be a half of the peers in the same depth, i.e., $X_d=E[\N{d}/2]$. The next proposition shows that when the arborescence up to depth $d^*$ is nearly a binary tree, the portion of flow-1 peers in each depth is concentrated on a half, $X_d/\N{d}\approx 1/2$:
\begin{proposition}\label{prop:colorOfPeers}
Conditioned on $\N{d}\geq 2 \phi_{d-1}\N{d-1}$ for all $d\leq d^*$, for any $\epsilon>0$ and sufficiently large $N$,
\begin{equation}
\prob{\left|\frac{X_d}{\N{d}}- \frac{1}{2}\right|<\epsilon, \forall d\leq d^*}\geq-\frac{2\log_2 N}{ \exp\left(\frac{(1-\epsilon)N^{1/3}}{4 \log_2^4N}\right)}.\NN
\end{equation}
\end{proposition}
\jour{The proof is provided in Appendix~D in \cite{Kim13tech2}.}
\tech{The proof is provided in Appendix~\ref{append:prop:colorOfPeers}.}

In this subsection, we have shown that the arborescence up to $d^*$ is nearly a binary tree and
a half of the peers in each depth are flow-1 peers with high probability. Since every flow-1 peer in the arborescence is  connected from the source only with  $d^*$ or less flow-1 edges, the flow-1 peers in the arborescence are connected from the source in flow graph~1 within $d^*$ hops. This means that the flow-1 peers can receive flow-1 chunks with delay $d^*=O(\log N)$. In the next subsection, we will consider how these chunks can be disseminated to the peers beyond $d^*$ through the flow-1 peers in depth $d^*$.

\subsection{Distance to the Remaining Peers}

We now show that the maximum distance from the flow-1 peers within depth $d^*$ to  peers $v$ with $d>d^*$ in flow graph~1.

Let $\setS{0}$ be the set of the peers within depth $d^*-1$  as shown in Fig.~\ref{fig:algorithmDescription}, i.e.,
$\setS{0}=\{v\in V\;|\; d(v)<d^* \}.$
Let $\setS{1}$ be the set of the flow-1 peers in depth $d^*$, i.e.,
$\setS{1}=\{v \in V\; | \; d(v)=d^*, \chi(v)=1\}.$
We have shown that the number of the peers in depth $d^*$ is larger than $\frac{N}{\log^c N}$ with high probability in Proposition~\ref{prop:mainExpansion} and a half of them are flow-1 peers that belong to $\setS{1}$ in Proposition~\ref{prop:colorOfPeers} with high probability. We then iteratively define $\setS{h}\subset V$ for $h>1$ and $m=1,2$ as follows:
$$\setS{h} = \{v \; | v \notin \cup_{l=0}^{h-1} S_l,\;\exists m \mbox{ s.t. } f_m(v)=1, p_m(v)\in S_{h-1}\}.$$
In other words, $\setS{h}$ is the set of the peers not in  $\cup_{l=0}^{h-1}\setS{l}$ that are connected from  the peers in $\setS{h-1}$ through  flow-1 incoming edges. (See the example in  Fig~\ref{fig:algorithmDescription}.  We also define $\setS{\leq h}\triangleq \cup_{d=0}^h\setS{d}$ and $\setS{>h} \triangleq \setS{\leq h}^C$. For simplicity, let $\Sd{h}\triangleq|\setS{h}|$, $\Sd{\leq h}\triangleq |\setS{\leq h}|$, and $\Sd{>h}\triangleq|\setS{>h}|$.

Since the peers in $\setS{h}$ are connected from the peers in $\setS{h-1}$ through flow-1 edges, the distance from the peers in $\setS{0}$ to the peers in $\setS{h}$ in flow graph~1 is exactly $h$. Since the peers in $\setS{1}$ are $d^*$ hops away from the source in flow graph 1, the peers in $\setS{h}$ must be  $d^*+h$ hops from the source in flow graph 1. Therefore, our goal in this subsection is to show that $\setS{\leq O(\log N)}$ covers almost all peers in the network with high probability.

While we have used edge expansion in the previous subsection, we use the contraction of the number of remaining peers. Specifically, define contraction ratio $\gamma_h$ as
$$\gamma_h=\frac{\SC{h}}{\SC{h-1}}.$$
From the result in the previous subsection, the initial contraction ratio $\gamma_1$ is upper bounded by
$$\gamma_1=\frac{\SC{1}}{\SC{0}}\leq 1-\frac{1}{\log^c N}.$$
The contraction ratio can be interpreted as the fraction by which the number of  peers not within $h$ hops from the source in flow graph~1 decreases at each  $h$. We will first show that the sequence of $\gamma_1, \gamma_2, \cdots,$ is a martingale in Proposition~\ref{prop:martingale}, and then establish the following concentration result:
$$\gamma_1 \approx \gamma_2 \approx \cdots \approx \gamma_h.$$
This concentration result leads to the following:
\begin{align}
\SC{h}=& \left(\prod_{l=1}^h \gamma_l\right) \SC{0} \approx\gamma_1^h \SC{0}\leq
\left(1-\frac{1}{\log^cN}\right)^h N.\label{eq:remainingPeers}
\end{align}
Taking $h=\log N$, the R.H.S. of (\ref{eq:remainingPeers}) approaches  $ N/\exp({\log^{1-c}N})$, which converges to zero as $N$ increases. This implies that in a large network, only $o(1)$ fraction of peers do not belong to $\setS{\leq\log N}$, hence all the remaining  peers beyond depth $d^*$ are within $d^*+\log N$ hops from the source in flow graph 1.

We  consider the distribution of $S_1, S_2, \cdots.$ using the example in Fig.~\ref{fig:algorithmDescription}. Suppose  the outgoing edges from the peers in $\setS{0}$ and $\setS{1}$ are given as in the figure, and thus $\setS{2}$ is also determined. We find the distribution of $S_3$ by drawing random edges from $\setS{2}$. Recall that each peer beyond depth $d^*$ assigns flow~1 to one of its incoming edges uniformly at random. Among the peers in $\setS{>2}$, let $\setS{>2,m}$ be the set of the peers beyond depth $d^*$ that assign flow~1 to its $m$-th incoming edge, i.e.,
$$\setS{>2, m}=\{v\in \setS{>2}\;|\; f_m(v)=1 \},$$
and let $\SC{2,m}=|\setS{>2, m}|$. When we draw a layer-1 edge from a peer in $\setS{2}$, there is a layer-1 edge outgoing from every peer in $\setS{\leq 1}$. Hence, there are $\SC{1}$ peers that we can draw the layer-1 edge to. Among these peers, if we draw the edge to a peer in $\cap \SC{2,1}$, i.e., a peer that assigned flow~1 to its first incoming edge, the peer in $\Sd{2}$ and the chosen peer are connected through flow-1 edge, and thus the chosen peer will be added to $\setS{3}$. If we draw the edge to a peer $\SC{2,2}$ that assigned flow~2 to its layer-2 incoming edge, the peer will not be  added to $\setS{3}$ because the chosen  peer is connected through flow-2 edge.  Hence, there are $\SC{1}$ peers that we can draw layer-1 edges from $\Sd{1}$ to and drawing to a peer in $\SC{2,1}$ will be considered to be a success. Hence, $\Sd{3,1}$ is a hyper geometric random variable with parameters $(\SC{1}, \SC{2,1}, \Sd{2})$.
By symmetry,  $\Sd{3,2}$ is a hyper geometric random variable with $(\SC{1}, \SC{2,2}, \Sd{2})$. In general, we can conclude the following:
\begin{lemma}\label{lemma:hyperS}
Conditioned on $\Sd{0}, \cdots, \Sd{h}$, and  $(\SC{h,1},\SC{h,2})$,
$\Sd{h+1,1}$ and $\Sd{h+1,2}$ are independent hyper geometric random variables with parameters $(\SC{h-1}, \SC{h,1}, \Sd{h} )$ and $(\SC{h-1}, \SC{h,2}, \Sd{h})$, respectively.
\end{lemma}

Using the properties of hyper geometric random variables, we can derive the following martingale property:
\begin{proposition}\label{prop:martingale}
The contraction ratios $\gamma_1, \gamma_2,\cdots$ form a martingale sequence, i.e.,
$$E[\gamma_{h+1}\;|\; \gamma_1, \gamma_2,\cdots \gamma_h]=\gamma_h.$$
\end{proposition}
\jour{The proof is provided in Appendix~E in \cite{Kim13tech2}.}
\tech{The proof is provided in Appendix~\ref{append:prop:martingale}.}

The result implies that the expected ratio $\SC{d}/\SC{d-1}$ of the remaining peers at each hop $d$ remains $\gamma_1$ as $d$ increases.
 We extend the martingale result to the following stronger concentration result:
\begin{proposition}\label{prop:concentrationOnContraction}
Conditioned on $S_0, S_1, \cdots, S_h,$ and $(\SC{h,1},\SC{h,2})$, for any $\epsilon>0$ and sufficiently large $N$,
\tred{$$\prob{\gamma_{h+1}<\gamma_h+\epsilon}\geq 1- 2\exp\left(-\frac{\epsilon^2 \min(\SC{h,1},\SC{h,2})^2}{\Sh}\right)$$}
\tblue{
$$\prob{\gamma_{h+1}\leq\gamma_h+\epsilon}\geq 1- 2\exp\left(-\frac{\epsilon^2 \SC{h}}{2}\right)$$}
\end{proposition}
\jour{The proof is provided in Appendix~F in \cite{Kim13tech2}.}
\tech{The proof is provided in Appendix~\ref{append:prop:concentrationOnContraction}.}

From Proposition~\ref{prop:concentrationOnContraction}, the contraction ratio only increases slightly with the number of hops with high probability. Hence, this concentration result holds for $\gamma_1, \gamma_2, \cdots, \gamma_h$, the numbers of  remaining peers,  $\SC{1}, \SC{2}, \cdots, \SC{h}$, decrease exponentially. In the next proposition, we show that  $\SC{h}$ becomes smaller than $N/\exp(\log^{1-c}N)$ before $h<\log N$:
\begin{proposition}\label{prop:CoverageToRemainginPeers}
Let $h^*$ be the first $h$ that satisfies $\SC{h}< N/\exp
( \log^{1-c} N   )$. Conditioned on \tred{$S_0>N/\log^{c}N$,}
\tblue{$S_1>N/\log^{c}N$,}
\tred{$$\prob{h ^* <\log N}\geq 1-2\log N\exp\left(-\frac{N}{\exp(\log^{1-c}N)} \right).$$}
\tblue{
\begin{align}
&\prob{h ^* \leq (1+\epsilon) \log N}\NN\\
&\geq 1-2(1+\epsilon)\log N\exp\left(-\frac{N^{1/2}}{2\exp(\log^{1-c}N)} \right).\NN
\end{align}
}
\end{proposition}
\jour{The proof is provided in Appendix~G in \cite{Kim13tech2}.}
\tech{The proof is provided in Appendix~\ref{append:prop:CoverageToRemainginPeers}.}

In the previous section, we have shown that flow-1 peers in $\setS{0}$
are within $d^*=O(\log N)$ hops away from the source in flow graph~1. Here, we have shown that ($1-o(1)$) fraction
 of peers beyond depth $d^*$ are within $\log N$ hops away from the flow-1 peers in $S_0$ in flow graph~1. Hence,
  all the aforementioned peers are within $O(\log N)$ hops in flow graph~1. In the next subsection, we show that the flow-2 peers in $\setS{0}$ are also connected in flow graph~1 from the source within $O(\log N)$ hops.

\subsection{Distance to Flow~2 Peers }

Previously, we have shown that all the flow-1 peers within depth $d^*$ and almost all peers below depth $d^*$ are connected from the source in flow graph~1 with high probability. In this subsection, we consider how the remaining peers, the flow-2 peers within depth $d^*$, are connected in flow graph~1. By showing that most of them are  connected from the peers below depth $d^*$ with a flow-1 incoming edge with high probability, we show that the flow-2 peers are also connected in flow graph~1 within $O(\log N)$ hops.

We provide the basic idea of the proof and refer to our online technical report for the details.
As we can see in Fig.~\ref{fig:algorithmDescription}, all the outgoing edges from flow-2 peers in $\setS{0}$ are all assigned flow~2. Recall that every peer has one flow-1 incoming edge under the RFA algorithm. Hence, all the flow-1 edges incoming to the flow-2 peers must come from the flow-1 peers in $\setS{0}$ or from the peers below depth $d^*$. However, we have shown in Proposition~\ref{prop:mainExpansion} that $\N{d}>2(1-\phi_d) \N{d-1}$ with high probability. This means that the number of the  flow-1 edges that begin at the flow-1 peers in depth $d-1$ and end at flow-2 peers within depth $d$ and the flow-2 edges that begin at flow-2 peers in depth $d-1$ and end at flow-2 peers within depth $d$ is no larger than $2\phi\N{d-1}$, which is a small fraction compared to all the edges outgoing from the peers in depth $2$. In Fig.~\ref{fig:algorithmDescription}, the second  edge outgoing from the second peer in depth~2 corresponds to this case. Overall, only a small fraction $2\max \phi_d$ of edges beginning at the peers within depth $d^*$ end at the peers with different flows. Hence, most flow-1 incoming edges of the flow-2 peers within depth $d^*$ come from the peers below depth $d^*$, most of which are $d^*+h^*$ hops away from the source in flow graph~1. Hence, we can conclude that almost all flow-2 peers in depth 2 are within $d^*+h^*+1$ hops away from the source in flow graph~1, where $
d^*+h^*+1=O(\log N)$.

\section{Simulation Results}\label{sec:simulationResults}
\begin{figure}[t]
\centering
\subfigure[Number of Disconnected Peers]{\label{fig:disconnectedPeers}
\includegraphics[width=0.9\mylengthNarrow]{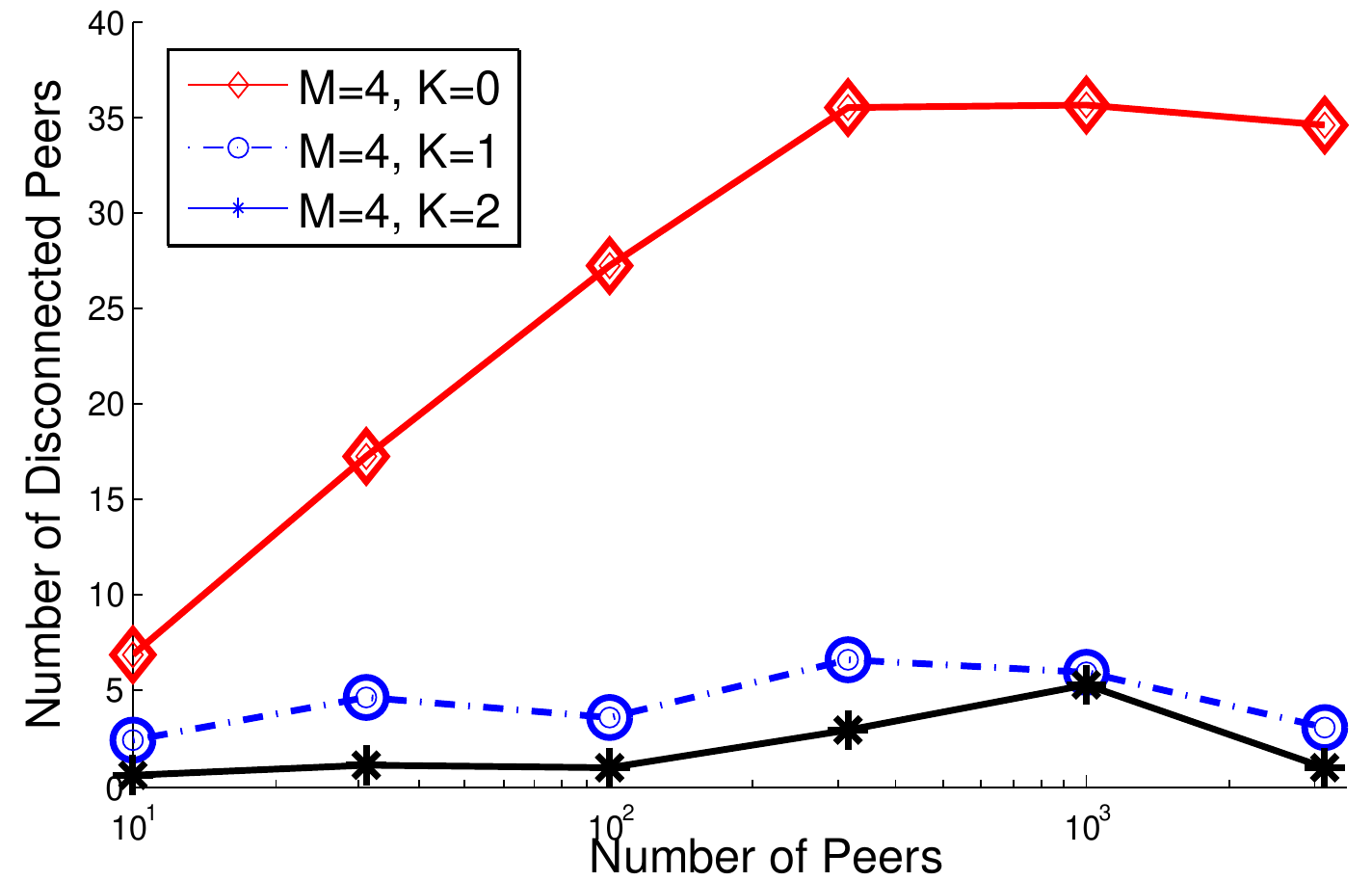}} 

\subfigure[The maximum delay]{\label{fig:MaxDistance}
\includegraphics[width=0.9\mylengthNarrow]{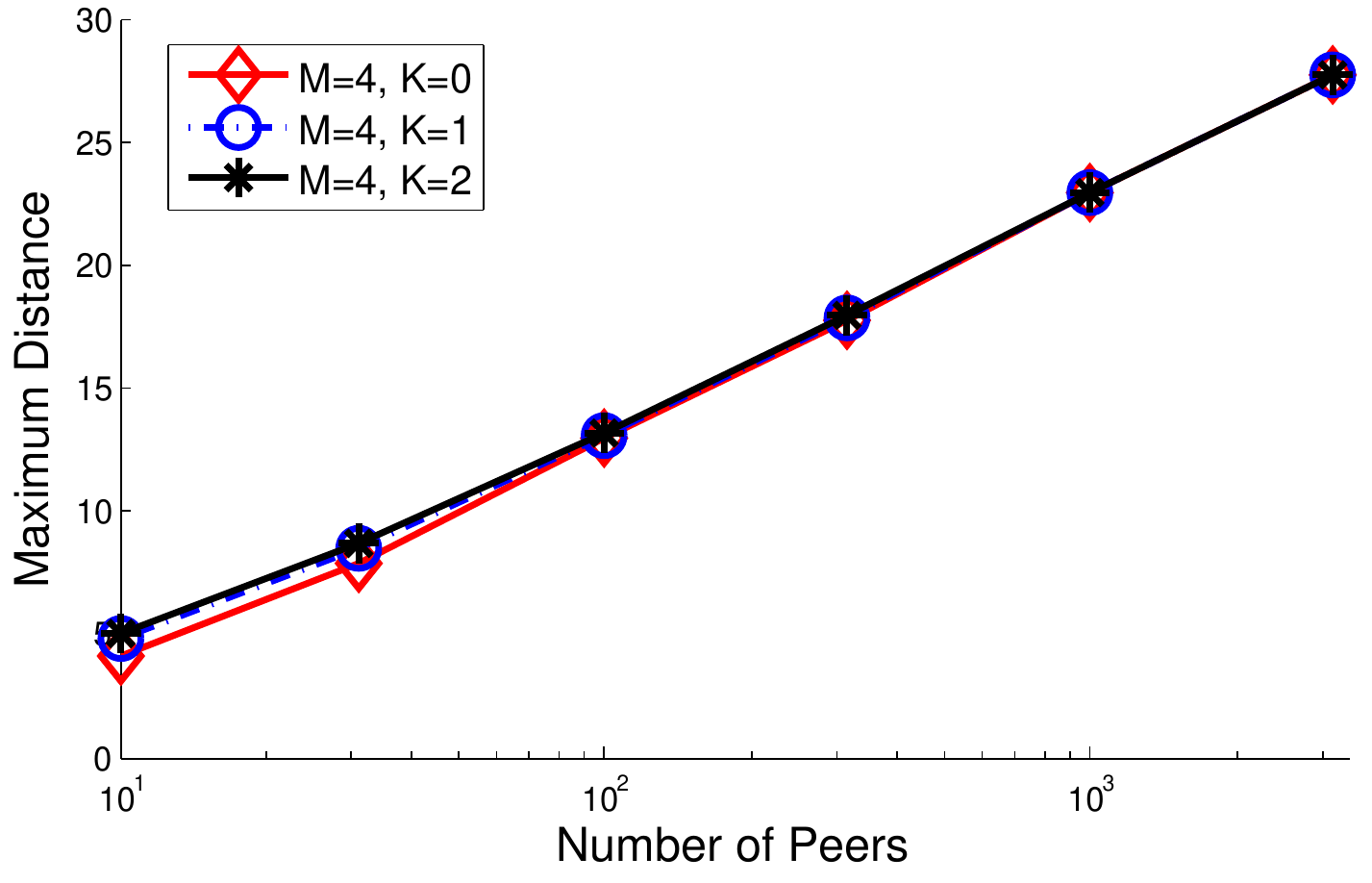}} 

\caption{Simulation results for  $N=10, 31, 100, 316, 1000, 3163$, $M=4$, and $K=0, 1, 2$, where $N$ is the number of the peers in the network,  $M$ is the number of layers, $K$ is the number of extra layers}
\label{fig:simulationResults}
\end{figure}

In the previous section, our P2P algorithm is proven to be able to  disseminate real-time video at the optimal rate to all peers, except $o(1)$ fraction of  peers,  with $O(\log N)$ delay.
We call the small fraction of peers disconnected peers. In practice, it is important to ensure that no peer is disconnected. Further, our results are asymptotic in nature, and it is important to ensure good performance even when $N$ is small. In this section, we address these issues, while a more comprehensive implementation can be found in \cite{Kim13tech3}.

Besides the $M$ layers, we assume that the peer-pairing algorithm constructs another layer, and thus each peer has an $(M+1)$-th parent and child. The chunk dissemination over the first $M$ layers is  exactly the same as before. If there are disconnected peers that cannot receive some flows, they request the missing flow from their $(M+1)$-th parents, if the parent is receiving the flow. For example, in Fig.~\ref{fig:flowgraph}, peers $i$ and $j$ that cannot receive flow-1 chunks request flow~1 from their $(M+1)$-th parents. If only one of them succeeds in receiving flow~1, then the other peer will also receive flow~1 through its flow-1 incoming edge. For this, the $(M+1)$-th parents of the disconnected peers should contribute additional unit bandwidth to their $(M+1)$-th outgoing edges. However, since the fraction of the disconnected peers is $o(1)$, only $o(1)$ fraction of peers should contribute the additional bandwidth. In practice, we may need more than one extra layer to make sure that there are no disconnected peers. In this section, we study the number of extra layers required through simulations.

We plotted the number of disconnected peers and the maximum depth in Fig~\ref{fig:simulationResults}, changing network size $N$,  the number $M$ of layers, and the number $K$ of extra layers. For each setting, we collected  simulation results from 100 samples and find the average value.
In Fig.~\ref{fig:disconnectedPeers}, we can observe that the number of disconnected peers dramatically decreases when  more extra layers are used. In particular, there are 5 or less disconnected peers regardless of the network size $N<3163$. The line with $(M=4, K=0)$ represents the number of disconnected peers under our original algorithm. When we use extra layers, these peers request missing flows from the parents in the extra layers, which requires  the parents to spend 25\% more bandwidth for each extra layer.  However, the number of such parents contributing more bandwidth is small (less than 35 for $N\leq 3163$) , and thus the impact of this additional bandwidth consumption to the network is negligible. In Fig.~\ref{fig:MaxDistance}, we have plotted the delay to all the connected peers as a function of $N$. The $x$ axis is in log scale and so the figure confirms the $O(\log N)$ bound on the delay. Further, it shows that even if the extended algorithm connects disconnected peers using extra layers, the maximum distance hardly increases.

\section{Conclusions}

We have presented an algorithm for transmitting streaming data in a peer-to-peer fashion over a network formed by the superposition of two random 1-regular digraphs. We showed that the algorithm can deliver chunks to all but a negligible fraction of nodes with a delay of $O(\log N)$ with high probability, where $N$ is the number of peers in the network. This improves upon an existing result where the constant in the $O(\log N)$ delay goes to infinity when the streaming rate approaches the capacity of the network.
The practical implications of the result have been studied in a real implementation in \cite{Kim13tech3}.

\appendix
\subsection{Proof of Lemma~\ref{lemma:2-ary}}
\label{append:lemma:2-ary}

To satisfy $\Nd=2\N{d-1}$, we must have $\Nd'=\Nd''=\N{d-1}$. We first
find the probability that  the $N_{trial}$ balls drawn without replacement from a jar containing $N_{all}$ balls including $N_{red}$ red balls are all red. Since the number $X$ of the red balls among the drawn balls is a hyper geometric random variable with ($N_{all}$, $N_{red}$, $N_{trial}$), we can derive the probability using \cite[Prob. (c), page 11]{Dubhashi2009}:
\begin{equation}
\prob{X=N_{trial}}=\frac{\comb{N_{red}}{N_{trial}}}{\comb{N_{all}}{N_{trial}}}\geq \left(\frac{N_{red}-N_{trial}+1}{N_{all}-N_{trial}+1}\right)^{N_{trial}}.\label{eq:hyperAllred}
\end{equation}
The second term can be seen as the number of possible ways to pick $N_{trial}$ red balls over the number of possible ways to pick any $N_{trial}$ balls.

Recall that $\Nd'$ is a hyper geometric random variable with ($\NC{d-2}$, $\NC{d-1}$, $\N{d-1}$). Applying $\Nd'$ to (\ref{eq:hyperAllred}), we have the following bound: conditioned on $\N{0}, \cdots, \N{d-1}$,
$$\prob{\Nd'=\N{d-1}}\geq
\left(\frac{\NC{d-1}-\N{d-1}+1}{\NC{d-1}+1}  \right)^{\N{d-1}}.$$
Conditioned on $\Nd'=\N{d-1}$, $\Nd''$ is also a hyper geometric random variable with ($\NC{d-2}$, $\NC{d-1}-\N{d-1}$, $\N{d-1}$). Hence, we also have
\begin{align}
&\prob{\Nd''=\N{d-1}\;|\; \Nd'=\N{d-1}}\NN\\
&\geq \left(\frac{(\NC{d-1}-\N{d-1})-\N{d-1}+1}{\NC{d-1}+1}  \right)^{\N{d-1}}.\NN
\end{align}
Combining both, we conclude the following bound: conditioned on
$\N{0}, \cdots, \N{d-1}$,
\begin{align}
&\prob{\Nd=2\N{d-1}}=\prob{\Nd'=\N{d-1}, \Nd''=\N{d-1}}\NN\\
&=\prob{\Nd''=\N{d-1} \;|\;\Nd'=\N{d-1}}\prob{\Nd'=\N{d-1}}\NN\\
&\geq \left( 1-\frac{2\N{d-1}}{\NC{d-1}+1}\right)^{2\N{d-1}}>1-\frac{4\N{d-1}^2}{\NC{d-1}+1}\NN.\NN
\end{align}
Note that if $\N{h}=2\N{h-1}$ for all $0<h\leq d$, it is easy to see that $\N{h}=2^h$ and $\N{\leq h}=2^{h+1}-1$ for all $h\leq d$. Using this, we derive the following probability:
\begin{align}
&\prob{\N{h}=2\N{h-1}, \forall h\leq d}\NN\\
&=\prod_{h=1}^{ d}\prob{\N{h}=2\N{h-1}\;|\;\N{l}=2\N{l-1},\forall l<h}\NN\\
&\geq\prod_{h=1}^{d}  \left( 1-\frac{4\N{h-1}^{2}}{N-\N{\leq h-1}+1}\right)
\geq  \left( 1-\frac{4\N{d-1}^{2}}{N-\N{\leq d-1}+1}\right)^d
\NN\\
&>
\left( 1-\frac{2^{2d}}{N-2^d}\right)^d>  1-\frac{d\cdot 2^{2d}}{N-2^d}\NN,
\end{align}
which is the result of this lemma.

\subsection{Proof of Proposition~\ref{prop:concentrationOnNd}}
\label{appex:prop:concentrationOnNd}

Since $\N{d+1}'$ is  a hyper geometric random variable with  $(\NC{d-1}, \NC{d}, \N{d})$, the following concentration result holds \cite[page 98]{Dubhashi2009}:
\begin{align}
&\prob{ |\Ndone'-E[\Ndone']|> t\;|\;\N{0},\cdots, \Nd}\NN\\
&\leq \exp{\left[-2 \frac{(\NC{d-1}-1) t^2}{\NC{d}(\Nd-1)}\right]}\leq \exp{\left[-2 \frac{\NC{d-1}t^2}{\NC{d}\Nd}\right]}\NN\\
&\leq \exp{\left[-2 \frac{t^2}{\Nd}\right]}\leq
\exp{[-2\sigma_d^2 N_d]} \NN\label{eq:hyper1}
\end{align}
where  $t= (1-\frac{\N{d}}{\NC{d-1}}-\phi)\Nd\geq    \sigma_d\Nd$. From (\ref{eq:mean_of_Nd'}), we can infer that if  $\Ndone'< \phi \Nd$, then $|\Ndone'-E[\Ndone']|>t$. Hence, we have
\begin{equation}
\prob{ \Ndone'< \phi \Nd\;|\;\N{0},\cdots, \Nd}
\leq \exp{[-2\sigma_d^2 N_d]}.\label{eq:concentration1}
\end{equation}

Similarly, we can derive the concentration result on $\Ndone''$, which is a hyper geometric random variable with $(\NC{d-1}, \NC{d}-\N{d+1}', \N{d})$:
\begin{align}
&\prob{ |\Ndone''-E[\Ndone'']|>s\;|\;\N{0},\cdots, \Nd, \Ndone'}\NN\\
&\leq \exp{\left[-2 \frac{(\NC{d-1})s^2}{(\NC{d})\Nd}\right]}\leq
\exp{[-2\sigma_d^2 N_d]} \NN
\end{align}
where $s=(1-\frac{\N{d}+\Ndone'}{\NC{d-1}}-\phi)\Nd
\geq \sigma_d \Nd$. Since $ |\Ndone''-E[\Ndone'']|>s$ if 
$\Ndone''< \phi \Nd$ from (\ref{eq:mean_of_Nd''}), we have
\begin{equation}
\prob{ \Ndone''< \phi \Nd\;|\;\N{0},\cdots, \Nd}
\leq \exp{[-2\sigma_d^2 N_d]},\label{eq:concentration2}
\end{equation}
where the upper bound is independent of $\Ndone'$.

Combining (\ref{eq:concentration1}) and (\ref{eq:concentration2}), we can conclude the following result: conditioned on $\N{0}, \N{1},\cdots, \N{d}$,
\begin{align}
&\prob{\Ndone\geq 2\phi \Nd}\geq \prob{\Ndone'\geq \phi \Nd, \Ndone''\geq \phi \Nd}\NN\\
&=\prob{\Ndone''\geq \phi \Nd|\Ndone'\geq \phi \Nd}
\prob{\Ndone'\geq \phi \Nd}\NN\\
&\geq 1-2\exp{[-2 \sigma_d^2 \Nd ]},\NN
\end{align}
which is the result of this proposition.

\subsection{Proof of Proposition~\ref{prop:mainExpansion}}
\label{append:prop:mainExpansion}

We fist show that if $\N{d}\geq 2\phi_{d-1}\N{d-1}$ for all $d\leq d^*$, then
$\N{d^*}\geq (1-\epsilon)N/\log^c N$ is automatically satisfied for sufficiently large $N$.  This allows us to focus on the probability of $\N{d}\geq 2\phi_{d-1}\N{d-1}$ for all $d\leq d^*$ to prove this proposition.

Suppose $\N{d}\geq 2\phi_{d-1}\N{d-1}$ for all $d\leq d^*$. We then have
\begin{align}
&\N{d^*}\geq (\prod_{d=0}^{d^*-1}2\phi_d)\N{0}=2^{d^*}\prod_{d=0}^{d_1-1} \phi_d \cdot\prod_{d=d_1}^{d_2-1}\phi_d \cdot\prod_{d=d_2}^{d^*-1}\phi_d\NN\\
&\geq \frac{N}{\log^c N}\left(1-N^{-\frac{1}{9}}\right)^{d_2-d_1}\left(\frac{\NC{d^*-1}}{\NC{d_2-1}} \right)^3\label{eq:lastOne}
\end{align}
Since $d_2=O(\log N)$ is much smaller than $N^{1/9}$,
 the second term in (\ref{eq:lastOne}) converges to 1 as $N$ increases. Since
\begin{equation}
\N{d}\leq 2^d \mbox{ and } \N{\leq d-1}\leq \prod_{h=0}^{d-1} 2^h=2^{d}-1,\label{eq:maximumWithin}
\end{equation} we have
$$\frac{\NC{d^*-1}}{\NC{d_2-1}} = \frac{N-\N{\leq d^*-1}}{N-\N{\leq d_2-1}}
\geq \frac{N-\N{\leq d^*-1}}{N}\geq 1-
\frac{2}{\log^c N}.$$
Hence, the third term in (\ref{eq:lastOne}) is lower bounded by $(1-2/\log^c N)^3$, which also converges to one. Thus, for any $\epsilon>0$ and sufficiently large $N$, we have $\N{d^*}\geq (1-\epsilon)\frac{N}{\log^c N}.$

To prove this proposition,  we only need to show
$$\prob{\N{d}\geq 2\phi_{d-1} \N{d-1},\forall d\leq d^*}\geq 1-\frac{2\log_2 N}{N^{2/3}}.$$
 Applying $d_1$ to Lemma~\ref{lemma:2-ary}, we can derive
\begin{align}
&\prob{\N{d}\geq 2\phi_{d-1} \N{d-1}, \forall d\leq d_1}\NN\\
&=\prob{\N{d}= 2 \N{d-1}, \forall d\leq d_1}\geq 1-\frac{d_1 2^{2d_1}}{N-2^{d_1}}\NN\\
&\geq 1-\frac{(\frac{1}{3}\log_2 N +1)\cdot 4 N^{\frac{2}{3}}}{N-2N^{\frac{1}{3}}}\geq 1-\frac{2\log_2 N}{N^{\frac{1}{3}}}\label{eq:largestProb}
\end{align}

For $d_1< d \leq d_2$, suppose $ \N{h}\geq 2\phi_{h-1}\N{h-1},\;\;\forall h<d$.
Since $2^{d-1}\leq2^{d_2-1}<N^{5/6}$ and $\N{d}>N_{d_1}=2^{d_1}\geq N^{1/3}$, we have the following probability bound from  Proposition~\ref{prop:concentrationOnNd}:
\begin{align}
&q(d)\triangleq\prob{\Nd< 2\phi_{d-1}\N{d-1}| \N{h}\geq 2\phi_{h-1}\N{h-1},\;\;\forall h<d}\NN\\
&\leq 2\exp{\left[ -2 \N{d-1} \left(N^{-\frac{1}{9}}- \frac{2\N{d-1}}{\NC{d-2}}\right)^2\right]}\NN\\
&\leq 2\exp{\left[ -2 \N{d-1} \left(N^{-\frac{1}{9}}- \frac{2\cdot 2^{d-1}}{N-2^{d-1}}\right)^2\right]}\NN\\
&\leq 2\exp{\left[-2N^{\frac{1}{3} }(N^{-\frac{1}{9}}(1-o(1)))^2\right]}\NN\\
&\leq 2\exp{\left[-2N^{\frac{1}{9} }(1-o(1))\right]}.\label{eq:finalBound}
\end{align}

For $d_2< d\leq d^*$, suppose $\N{h}\geq 2\phi_{h-1}\N{h-1}$ for all $h<d$. From Proposition~\ref{prop:concentrationOnNd}, we have
\begin{align}
&q(d)
\leq 2\exp{\left[ -2 \N{d-1} \left(1-\frac{2\N{d-1}}{\NC{d-2}}-\left( \frac{\NC{d-1}}{\NC{d-2}}\right)^3\right)^2\right]}\NN\\
&= 2\exp{\left[ -2  \N{d-1}\left(\frac{\N{d-1}}{\NC{d-2}}-3\left(\frac{\N{d-1}}{\NC{d-2}}\right)^2+\left(\frac{\N{d-1}}{\NC{d-2}}\right)^3\right)^2\right]}\NN\\
&= 2\exp{\left[ -2  \frac{\N{d-1}^3}{\NC{d-2}^2}\left(1-3\frac{\N{d-1}}{\NC{d-2}}+\left(\frac{\N{d-1}}{\NC{d-2}}\right)^2\right)^2\right]}\NN\\
&\leq 2\exp{\left[ -2  \frac{\N{d-1}^3}{\NC{d-2}^2}(1-o(1))\right]}\NN\\
&\leq
2\exp{\left[ -\frac{2 \N{d-1}^3}{N^2}(1-o(1))\right]}\label{eq:upperbound3}
\end{align}
for sufficiently large $N$ because
$$\frac{\N{d-1}}{\NC{d-2}}\leq \frac{2^{d-1}}{N-2^{d-1}}\leq \frac{2^{d^*-1}}{N-2^{d^*-1}}=o(1).$$
Note that
\begin{align}
\N{d-1}&\geq \N{d_2}=2^{d_2-d_1}\phi_{d_2-1}\cdots\phi_{d_1}\N{d_1}\NN\\
&\geq 2^{d_2}(1-N^{-1/9})^{d_2-d_1}\geq N^{5/6}(1-o(1)).\NN
\end{align}
Using this, the probability in (\ref{eq:upperbound3}) is upper bounded by
\begin{align}
q(d)\leq2 \exp{\left[-2(1-o(1)) N^{\frac{1}{2}}\right]}.\label{eq:upperboun4}
\end{align}

Combining (\ref{eq:largestProb}), (\ref{eq:finalBound}), and (\ref{eq:upperboun4}), we can conclude that $\Nd< \phi_{d-1}\N{d-1}$ for $d\leq d^*$ with high probability:
\begin{align}
&\prob{\Nd\geq  2\phi_{d-1}\N{d-1},\;\forall d\leq d^*}\NN\\
=&\prob{\Nd\geq  2\phi_{d-1}\N{d-1},\;\forall d\leq d_1}\cdot\NN\\
&\prod_{d=d_1+1}^{d^*}\prob{\Nd\geq \phi_{d-1}\N{d-1}| \N{h}\geq 2\phi_{h-1}\N{h-1},\;\forall h<d}\NN\\
=&\prob{\Nd\geq  2\phi_{d-1}\N{d-1},\;\forall d\leq d_1}
\prod_{d=d_1+1}^{d_2}q(d)\prod_{d=d_2+1}^{d^*}q(d)\NN\\
\geq&\left(1-\frac{2\log_2 N}{N^{\frac{1}{3}}}\right)\left(1-2\exp{\left[-2N^{\frac{1}{9} }(1-o(1))\right]}\right)^{d_2-d_1}\cdot\NN\\
&
\left(1-2\exp{\left[-2 N^{\frac{1}{2}}(1-o(1))\right]}\right)^{d^*-d_2}\NN\\
\geq& 1-\frac{2\log_2 N}{N^{\frac{2}{3}}}(1+o(1)),\NN
\end{align}
which is equivalent to the result of this proposition.

\subsection{Proof of Proposition~\ref{prop:colorOfPeers}}
\label{append:prop:colorOfPeers}

Throughout this proof, we  fix $\N{0}, \cdots, \N{d^*}$ such that $\N{d}\geq 2 \phi_{d-1} \N{d-1}$ for $d\leq d^*$. Conditioned on these fixed $\N{d}$'s, we show that approximately a half of the peers in each depth are flow-1 peers.

We have shown in Lemma~\ref{lemma:2-ary} that the arborescence up to depth $d_1$ is binary. Hence,
 the number of flow-1 peers is the same as that of flow-2 peers in each depth $d\leq d_1$. Therefore, $X_d=\N{d}/2$ for all $d\leq d_1$.

\begin{figure}[ht]
\centering
\subfigure[Before drawing incoming edges to the peers in depth $d$]{\label{fig:colorOfPeers1}
\includegraphics[width=0.9\mylengthNarrow]{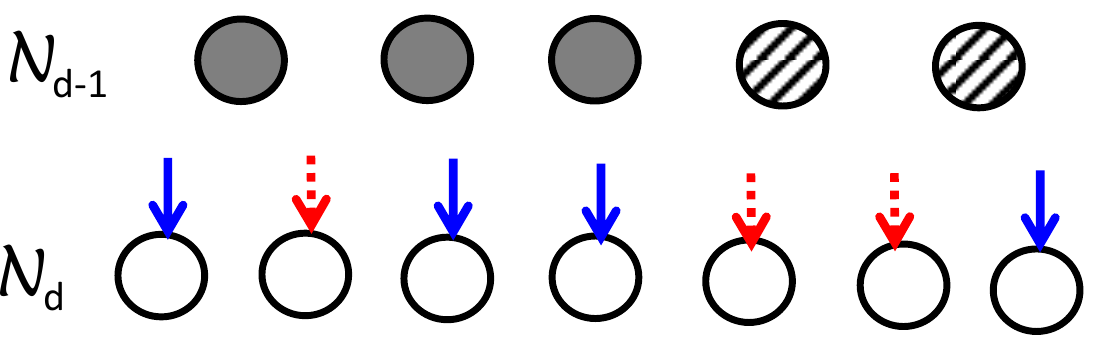}} 

\subfigure[After drawing layer-1 edges incoming to the peers in  $\setN{d,1}$]{\label{fig:colorOfPeers2}
\includegraphics[width=0.9\mylengthNarrow]{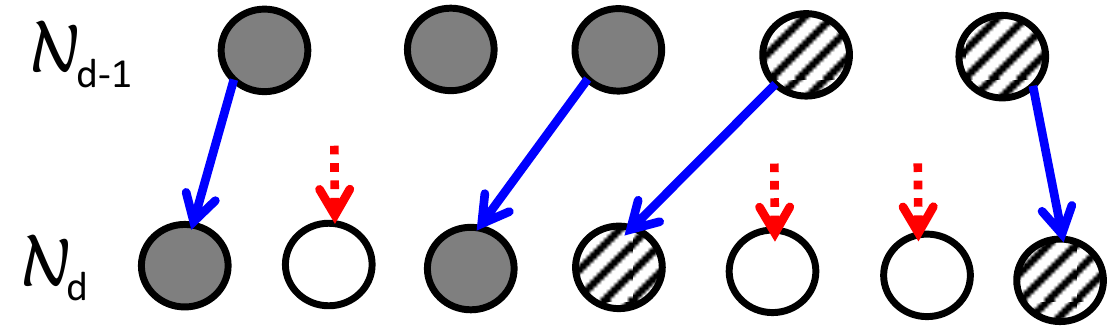}} 

\subfigure[After drawing layer-2 edges incoming to the peers in  $\setN{d,2}$]{\label{fig:colorOfPeers3}
\includegraphics[width=0.9\mylengthNarrow]{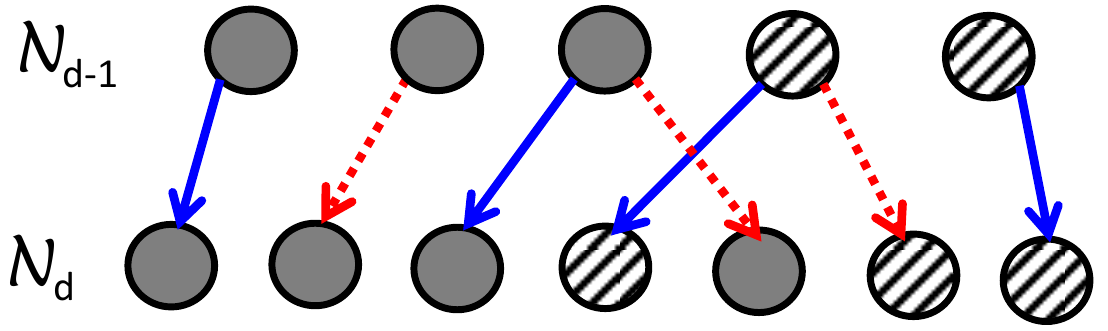}} 

\caption{Drawing random edges incoming to the peers in depth 2: each peer in depth $d$ inherit the main flow $\chi(v)$ from the parent in depth $d-1$: the solid (dotted) lines represent layer-1 (layer-2) edges. Gray peers represent flow-1 peers ($\chi(v)=1$) while the shaded peers represent flow-2 peers ($\chi(v)=2$).}
\end{figure}

For $d_1<d\leq d^*$, suppose that $X_1, \cdots, X_{d-1}$ are given. Under the RFA algorithm, every peer $v$ within depth $d^*$ determines $m^*(v)$ and inherits $\chi(v)$ from the parent over the $m^*(v)$-th incoming edge. Let $\setN{d,m}$ be the set of peers $v$ in $\setN{d}$ that are connected from the peers in $\setN{d-1}$ with layer-$m$ ($m$-th) edges, i.e., $\setN{d,m}=  \{ v \in \setN{d} \; |\; m^*(v)=m\}$, and let $\N{d,m}=|\setN{d,m}|$. In the example in Fig~\ref{fig:colorOfPeers1}, there are four peers in  $\setN{d,1}$ and three peers in $\setN{d,2}$. Note that the layer-1 incoming edge to a peer in  $\setN{d,1}$ can begin at any peer in $\setN{d-1}$ equally likely. Further,  two layer-1 edges cannot begin at the same peer. Hence, choosing the starting points of the layer-1 incoming edges to $\setN{d,1}$ is equivalent to choose $\N{d,1}$ peers from $\setN{d-1}$ without replacement. Hence, the number  $X_{d,1}$ of the layer-1 edges that begin at flow-1 (gray) peers in $\setN{d-1}$ is a hyper geometric random variable with parameter ($\N{d-1}, X_{d-1}, \N{d,1}$). Since drawing layer-2 edges incoming to $\setN{d,2}$ is independent of drawing layer-1 edges,  the number  $X_{d,2}$ of the layer-2 edges that begin at flow-1 (gray) peers in $\setN{d-1}$ is a hyper geometric random variable with parameter ($\N{d-1}, X_{d-1}, \N{d,2}$). Hence, we can find the following concentration result from \cite[page 98]{Dubhashi2009}:
\begin{align}
&\prob{ \left|X_{d,m}-\frac{X_{d-1}}{\N{d-1}}\N{d,m} \right|>\epsilon_{d-1} \N{d,m} \;|\;X_h, \forall h<d } \NN\\
&\leq \exp\left(- \frac{2(\N{d-1}-1)\left(\epsilon_{d-1} \N{d,m} \right)^2}{(\N{d-1}-\N{d,m})(\N{d,m}-1)}
 \right)\NN\\
 &\leq \exp\left(- \frac{2\N{d-1}
 \epsilon_{d-1}^2 \N{d,m} }{(\N{d-1}-\N{d,m})}
  \right)\label{eq:colorConcentration}
\end{align}
Since $\N{d}\geq 2\phi_{d-1}\N{d-1}$, we have
\begin{align}
&\N{d,1}+\N{d,2}\geq 2 \phi_{d-1}\N{d-1}\NN\\
\Leftrightarrow& \N{d,2}-(2\phi_{d-1}-1)\N{d-1}
\geq \N{d-1}-\N{d,1}\NN\\
\Rightarrow &\N{d-1}-(2\phi_{d-1}-1)\N{d-1}
\geq \N{d-1}-\N{d,1}\NN\\
\Leftrightarrow& 2(1-\phi_{d-1})\N{d-1}
\geq \N{d-1}-\N{d,1}.\label{eq:upper_phi1}
\end{align}
Similarly, we have
\begin{align}
&\N{d,1}+\N{d,2}\geq 2 \phi_{d-1}\N{d-1}\NN\\
\Leftrightarrow& \N{d,1}\geq 2\phi_{d-1}\N{d-1}- \N{d,2}\NN\\
\Rightarrow & \N{d,1}\geq (2\phi_{d-1}-1)\N{d-1}.\label{eq:upper_phi2}
\end{align}
Since the minimum of $\phi_{d}$ for $d<d^*$ is attained at $d^*-1$, we have
\begin{equation}
\phi_{d-1}\geq \phi_{d^*-1}\geq  \left(1-\frac{2^{d^*-1}}{N-2^{d^*-1}}\right)^3\geq
1-\frac{3}{\log^c N-1}\label{eq:upper_phi3}.
\end{equation}
Further, we have $\N{d-1}\geq \N{d_1}\geq N^{1/3}$.
Applying this, (\ref{eq:upper_phi1}), (\ref{eq:upper_phi2}), and (\ref{eq:upper_phi3}) to (\ref{eq:colorConcentration}), we have the following bound:
\begin{align}
&\prob{ \left|\frac{X_{d}}{\N{d}}-\frac{X_{d-1}}{\N{d-1}} \right|>\epsilon_{d-1} \;|\;X_h, \forall h<d   } \NN\\
&\leq\prob{ \left|\frac{X_{d,m}}{\N{d,m}}-\frac{X_{d-1}}{\N{d-1}} \right|>\epsilon_{d-1},\mbox{ for some }m \;|\;X_h, \forall h<d   } \NN\\
 &\leq 2\exp\left(- \frac{2
  \epsilon_{d-1}^2 (1-o(1))N^{1/3}\log^cN}{3}
   \right).\label{eq:colorConcentration3}
  \end{align}

Taking $\epsilon_d=1/\log_2^{2+c/2}N$ and $\cA{h}=\left|\frac{X_{h}}{\N{h}}-\frac{X_{h-1}}{\N{h-1}} \right|\leq\epsilon_{h-1}$. Then, we have the following
\begin{align}
&\prob{ \cA{1}, \cdots, \cA{d^*}}=\prod_{d=d_1+1}^{d^*}
\prob{\cA{d}\;|\;\cA{h}, \forall h<d}\NN\\
&\geq 1-2\log_2 N\exp\left(- \frac{2
  (1-o(1))N^{1/3}}{3 \log_2^4 N}
   \right). \label{eq:finalProbability}  \end{align}
Recall that $X_{d_1}/\N{d-1}=1/2$ since the arborescence is binary up to depth $d_1$. Hence, if $\cA{d}$ is true for all $d\leq d^*$, we have
$$\left|\frac{X_{d}}{\N{d}}-\frac{1}{2} \right|\leq \sum_{h=d_1}^{d-1}\epsilon_{h}\leq
\frac{d^*}{\log_2^{2+c/2}N}=o(1).$$
Hence, for any $\epsilon>$, there exists $N'$ such that for $N>N'$, $\left|\frac{X_{d}}{\N{d}}-\frac{1}{2} \right|\leq\epsilon$ for all $d\leq d^*$ and with the probability  in (\ref{eq:finalProbability}), which is the result of this proposition.

\subsection{Proof of Proposition~\ref{prop:martingale}}
\label{append:prop:martingale}
To consider a probability conditioned on $\Sd{0}, \Sd{1}, \cdots, \Sd{h}$, and $(\SC{h,1}, \SC{h,2})$, we first fix these variables. With these fixed values, $\Sd{h+1,m}$ is a hyper geometric random variable with $(\SC{h-1}, \SC{h,m}, \SC{h})$ by Lemma~\ref{lemma:hyperS}. Hence, the mean of $\Sd{h+1,m}$ is given by
$$E[\Sd{h+1,m}]=\frac{\SC{h,m} \SC{h}}{\SC{h-1}}\;\;\mbox{ for }m=1,2.$$
Thus, summing the means of $\Sd{h+1,1}$ and $\Sd{h+1,2}$, we have
\begin{align}
&E[\Sd{h+1}]=\frac{\Sd{h}(\SC{h,1}+\SC{h,2})}{\SC{h-1}}=\frac{\Sd{h}\SC{h}}{\SC{h-1}}.\NN
\end{align}
Since $\SC{h+1}=\SC{h}-\Sd{h+1}$, we can conclude
\begin{align}
E[\SC{h+1}]= \SC{h}-E[\Sd{h+1}] = \SC{h}(1-\frac{\Sd{h}}{\SC{h-1}})=\frac{\SC{h}^2}{\SC{h-1}}.\NN
\end{align}
If we divide both sides by $\SC{h}$, we have the following result:
$$E[\gamma_{h+1}|\Sd{0}, \Sd{1}, \cdots, \Sd{h}, \SC{h,1}, \SC{h,2}]=\gamma_{h}.$$
Since $\Sd{0}, \Sd{1}, \cdots, \Sd{h}$, and $(\SC{h,1}, \SC{h,2})$ fully determines $\gamma_1,\cdots, \gamma_h$, we can conclude the result in the proposition.

\subsection{Proof of Proposition~\ref{prop:concentrationOnContraction}}
\label{append:prop:concentrationOnContraction}

By definition of $\gamma_h$, we can derive the following:
\begin{align}
&\gamma_{h+1} > \gamma_h+\epsilon\NN\\
\Leftrightarrow&\frac{\SC{h}-\Sd{h+1}}{\SC{h}}> \frac{\SC{h}}{\SC{h-1}}+\epsilon\NN\\
\Leftrightarrow&\SC{h}-\Sd{h+1}>
\SC{h}
\left(  1-\frac{\Sd{h}}{\SC{h-1}}\right)+\epsilon \SC{h}  \NN\\
\Leftrightarrow & \Sd{h+1}<E[\Sd{h+1}]-\epsilon \SC{h}\NN\\
\Rightarrow&\Sd{h+1,m}<E[\Sd{h+1,m}]-\frac{\epsilon \SC{h}}{2}\;\;\mbox{ for some }m=1,2\NN
\end{align}
Using the union bound, we can rewrite the probability of $\gamma_{h+1} > \gamma_h+\epsilon$ as follows:
\begin{align}
&\prob{\gamma_{h+1} > \gamma_h+\epsilon}\leq \sum_{m=1}^2 \prob{\Sd{h+1,m}<E[\Sd{h+1,m}]-\frac{\epsilon \SC{h}}{2}}.\label{eq:probGamma}
\end{align}
Since $\Sd{h+1,m}$ is a hyper geometric random variable with $(\SC{h-1}, \SC{h,m}, \Sd{h})$, we can derive the following bound from \cite[page 98]{Dubhashi2009}
\tred{\begin{align}
&\prob{\Sd{h+1,m}<E[\Sd{h+1,m}]-\frac{\epsilon \SC{h}}{2}}\NN\\
&\leq \exp\left(-\frac{2 (\SC{h-1}-1)(\epsilon\SC{h}/2)^2}{(\SC{h-1}-\Sd{h})(\Sd{h}-1)}  \right)\NN\\
&\leq \exp\left(-\frac{ \epsilon^2\SC{h-1}\SC{h}^2}{2\SC{h}\Sd{h}}  \right)\NN\\
&\leq \exp\left(-\frac{ \epsilon^2\SC{h}^2}{2\Sd{h}} \right)
\leq \exp\left(-\frac{ \epsilon^2\min(\SC{h,1},\SC{h,2})^2}{\Sd{h}}
 \right).\NN
\end{align}
Applying the above to (\ref{eq:probGamma}), we obtain the result of this proposition. We can also rewrite this probability as
\begin{align}
&\prob{\Sd{h+1,m}<E[\Sd{h+1,m}]-\frac{\epsilon \SC{h}}{2}}\NN\\
&\leq \exp\left(-\frac{ \epsilon^2\SC{h-1}\SC{h}^2}{2\SC{h}\Sd{h}}  \right)\leq \exp\left(-\frac{ \epsilon^2\SC{h}}{2}  \right),\label{eq:mamma}
\end{align}
which will be used later.
}
\tblue{\begin{align}
&\prob{\Sd{h+1,m}<E[\Sd{h+1,m}]-\frac{\epsilon \SC{h}}{2}}\NN\\
&\leq \exp\left(-\frac{2 (\SC{h-1}-1)(\epsilon\SC{h}/2)^2}{(\SC{h-1}-\Sd{h})(\Sd{h}-1)}  \right)\NN\\
&\leq \exp\left(-\frac{ \epsilon^2\SC{h-1}\SC{h}^2}{2\SC{h}\Sd{h}}  \right)\leq \exp\left(-\frac{ \epsilon^2\SC{h}}{2}  \right).\NN
\end{align}
Applying the above to (\ref{eq:probGamma}), we obtain the result of this proposition.}

\subsection{Proof of Proposition~\ref{prop:CoverageToRemainginPeers}}
\label{append:prop:CoverageToRemainginPeers}

To prove this proposition, we define the following events:
\begin{align}
&\cA{h}=\{ \gamma_h \leq  \gamma_{h-1}+N^{-1/4} \}\mbox{ and }
\cB{h}=\{ h^*=h\}.\NN
\end{align}
For $h_1=\lfloor (1+\epsilon)\log N \rfloor$,
The probability $\prob{h^*\leq h_1 }$ can be expressed as follows, which we will explain afterwards:
\begin{align}
&\prob{h^*\leq h_1} = \prob{\cup_{h=1}^{h_1} \cB{h}}= \sum_{h=1}^{h_1}\prob{\cB{h}, \cB{l}^C, \forall 0<l<h}\NN\\
&\geq\sum_{h=1}^{h_1}\prob{\cA{h},\cB{h}, \cA{l},\cB{l}^C, \forall 0<l<h}\NN\\
&\geq \prod_{h=1}^{h_1}\prob{\cA{h}\;|\;\cA{l}, \cB{l}^C, \forall l<h}-\prob{\cA{l},\cB{l}^C, \forall l\leq h_1}\label{eq:iterative}\\
&=\prod_{h=1}^{h_1}\prob{\cA{h}\;|\;\cA{l}, \cB{l}^C, \forall l<h}.\label{eq:iterative2}
\end{align}

We first show that the inequality in (\ref{eq:iterative}) holds using induction. For $h_1=1$, (\ref{eq:iterative}) is satisfied because
$$\prob{\cA{1},\cB{1}}=\prob{\cA{1}}-\prob{\cA{1},\cB{1}^C}.$$
We now suppose (\ref{eq:iterative}) is satisfied for $h_1=h'-1$. Using this induction hypothesis, we can derive
\begin{align}
&\prod_{h=1}^{h'}\prob{\cA{h}\;|\;\cA{l}, \cB{l}^C, \forall l<h}\NN\\
\leq&\Bigg[
	\sum_{h=1}^{h'-1}\prob{\cA{h},\cB{h}, \cA{l},\cB{l}^C, \forall 0<l<h}\;\NN\\
&\;\;+ \prob{ \cA{l},\cB{l}^C, \forall l\leq h'-1}    \Bigg]\cdot\prob{\cA{h'}\;|\;\cA{l}, \cB{l}^C, \forall l<h'}\NN\\
\leq &
	\sum_{h=1}^{h'-1}\prob{\cA{h},\cB{h}, \cA{l},\cB{l}^C, \forall 0<l<h}\NN\\
&+	\prob{\cA{h'}, \cA{l},\cB{l}^C, \forall 0<l<h'}\NN\\
=& \sum_{h=1}^{h'}\prob{\cA{h},\cB{h}, \cA{l},\cB{l}^C, \forall 0<l<h}\NN\\
&\;\;+	\prob{\cA{h'}, \cB{h'}^C, \cA{l},\cB{l}^C, \forall 0<l<h'},\NN
\end{align}
which implies that (\ref{eq:iterative}) is also satisfied for $h_1=h'$. By induction, (\ref{eq:iterative}) is satisfied for all $h$.

We next show that (\ref{eq:iterative2}) is true, i.e., the second term in (\ref{eq:iterative}) is zero for sufficiently large $N$. Since $\cA{l}$ is satisfied for all $l\leq h_1$, we have $\gamma_h\leq \gamma_{h-1}+N^{-1/4}$ for all $h\leq h_1$. This implies that  for any $\epsilon'>0$,
\begin{align}
\gamma_h& \leq\gamma_1+\frac{h-1}{N^{1/3}}\leq \gamma_1+\frac{h_1}{N^{1/4}}\leq 1-\frac{\Sd{1}}{N-\Sd{0}}+\frac{h_1}{N^{1/4}}\NN\\
&\leq 1-\frac{1}{\log^c N}+\frac{h_1}{N^{1/4}}\leq 1-\frac{1}{(1+\epsilon')\log^c N}\NN
\end{align}
for sufficiently large $N$.
From (\ref{eq:remainingPeers}), we have
\begin{align}
\SC{h_1}&=\SC{0}\prod_{h=1}^{h_1}\gamma_h \leq
N\left( 1-\frac{1}{(1+\epsilon')\log^c N} \right)^{(1+\epsilon)\log N}
\NN\\
&=N\left[\left( 1-\frac{1}{(1+\epsilon')\log^c N} \right)^{(1+\epsilon)\log^c N}\right]^{\log^{1-c}N}\label{eq:sufficientLargeN}
\end{align}
for sufficiently large $N$. Note that if we take $\epsilon'<\epsilon$, the term in the bracket in (\ref{eq:sufficientLargeN}) is upper bounded by $e^{-1}$ for sufficiently large $N$. Hence, if $\cA{l}$ is true for all $l\leq h_1$, then $\SC{h_1}\leq N/\exp(\log^{1-c}N)$, which implies one of $\cB{1},\cB{2},\cdots, \cB{h^*}$ is true for large $N$. Hence, the second term in (\ref{eq:iterative}) is zero.

We finally expand the probability in (\ref{eq:iterative2}).
From \tblue{Proposition~\ref{prop:concentrationOnContraction},}
\tred{(\ref{eq:mamma}),} we have find the probability of $\cA{h}$ conditioned on $\Sd{0},\Sd{1},\cdots, \Sd{h-1}, \SC{h,1},$ and $\SC{h,2}$ by replacing $\epsilon=N^{-1/4}$.
Since $(\Sd{0},\Sd{1},\cdots, \Sd{h-1})$ fully determines whether $\cA{l}$ and $\cB{l}$ are true or not for $l<h$, we have
\begin{align}
\prob{\cA{h}^C\;|\;\cA{l}, \cB{l}^C, \forall l<h}\leq 2\exp\left(-\frac{N^{1/2} }{2\exp(\log^{1-c}N)}\right),\label{eq:applyThisFinal}
\end{align}
where the last inequality has been obtained using $\SC{h}>N/\exp(\log^{1-c}N)$ when $\cB{l}^C$ is true for $l<h$. Applying (\ref{eq:applyThisFinal}) to (\ref{eq:concentration2}), we can conclude
$$\prob{h^*\leq h_1}\geq 1- 2(1+\epsilon)\log N\exp\left(-\frac{N^{1/2} }{2\exp(\log^{1-c}N)}\right).$$

\bibliographystyle{IEEEtran}
\bibliography{joohwan}

\end{document}